%% file: main.tex
\newif\ifcomments   
\newif\ifanon       
\newif\ifcrypto     
\newif\ifllncs      
\newif\ifpublish    

\commentstrue

\ifllncs
  \documentclass[runningheads,a4paper]{llncs}

\else
  \documentclass[11pt,pdfa]{article}
  \usepackage[in]{fullpage}
\fi

\usepackage{iftex}
\ifPDFTeX
  \usepackage[utf8]{inputenc}
  \usepackage[noTeX]{mmap}
  \usepackage[T1]{fontenc}
\fi
\ifLuaTeX
  \usepackage{luatex85}
  \usepackage[noTeX]{mmap}
\fi

\usepackage{comment}
\usepackage[shortlabels]{enumitem}
\usepackage{bm}
\usepackage{amsfonts}
\usepackage{amsmath}
\usepackage{amssymb}
\usepackage{amsthm}
\usepackage{xcolor}
\usepackage{graphicx}
\usepackage{qtree}
\usepackage{tree-dvips}
\usepackage{float}
\usepackage{hyperref}
\usepackage[nameinlink]{cleveref}
\usepackage{braket}
\usepackage{mathrsfs}
\usepackage{tikz}
\usepackage{qcircuit}
\usepackage{framed} 
\usepackage{mdframed}

\allowdisplaybreaks 

\ifllncs
  
\else
  \hypersetup{colorlinks={true},linkcolor={blue},citecolor=magenta}

  \theoremstyle{plain}
  \newtheorem{theorem}{Theorem}[section]
  \newtheorem{lemma}[theorem]{Lemma}
  \newtheorem{claim}[theorem]{Claim}
  \newtheorem{corollary}[theorem]{Corollary}
  \newtheorem{definition}[theorem]{Definition}
  \newtheorem{remark}[theorem]{Remark}
  \newtheorem{construction}[theorem]{Construction} 

\usetikzlibrary{arrows.meta}
\usepackage{mathtools}

  \usepackage[style=alphabetic,minalphanames=3,maxalphanames=4,maxnames=99,backref=true]{biblatex}

  \DeclareFieldFormat{eprint:iacr}{Cryptology ePrint Archive: \href{https://ia.cr/#1}{\texttt{#1}}}
  \DeclareFieldFormat{eprint:iacrarchive}{Cryptology ePrint Archive: \href{https://eprint.iacr.org/archive/#1}{\texttt{#1}}}
  \AtEveryBibitem{
     \clearlist{address}
     \clearfield{date}
     \clearfield{isbn}
     \clearfield{issn}
     \clearlist{location}
     \clearfield{month}
     \clearfield{series}
     \ifentrytype{book}{}{
      \clearlist{publisher}
      \clearname{editor}
     }
  }
  
  \newcommand{\email}[1]{\href{mailto:#1}{\texttt{#1}}}
\fi
\newtheorem{fact}[theorem]{Fact}

\ifdefined\authorNote
\newcommand{\YTnote}[1]{{\color{blue} [{Yao-Ting:} #1]}}
\newcommand{\pnote}[1]{{\color{red} [Prab: #1]}}
\newcommand{\hnote}[1]{{\color{brown} \bf [Henry: #1]}}
\else
\newcommand{\YTnote}[1]{}
\newcommand{\pnote}[1]{}
\newcommand{\hnote}[1]{}
\fi

\input{headers}

\usepackage[style=alphabetic,minalphanames=3,maxalphanames=4,maxnames=99,backref=true]{biblatex}

  \DeclareFieldFormat{eprint:iacr}{Cryptology ePrint Archive: \href{https://ia.cr/#1}{\texttt{#1}}}
  \DeclareFieldFormat{eprint:iacrarchive}{Cryptology ePrint Archive: \href{https://eprint.iacr.org/archive/#1}{\texttt{#1}}}
\date{}
\title{Pseudorandom Strings from Pseudorandom Quantum States}
\author{Prabhanjan Ananth\thanks{\email{prabhanjan@cs.ucsb.edu}}\\ {\small UCSB} \and Yao-Ting Lin\thanks{\email{yao-ting\_lin@ucsb.edu}}\\ {\small UCSB} \and Henry Yuen\thanks{\email{hyuen@cs.columbia.edu}}\\ {\small Columbia University}}

\begin{document}

\maketitle

\begin{abstract}

\noindent We study the relationship between notions of pseudorandomness in the quantum and classical worlds. Pseudorandom quantum state generator (PRSG), a pseudorandomness notion in the quantum world, is an efficient circuit that produces states that are computationally indistinguishable from Haar random states. PRSGs have found applications in quantum gravity, quantum machine learning, quantum complexity theory, and quantum cryptography. Pseudorandom generators, on the other hand, a pseudorandomness notion in the classical world, is ubiquitous to theoretical computer science.  While some separation results were known between PRSGs, for some parameter regimes, and PRGs, their relationship has not been completely understood. 
\par In this work, we show that a natural variant of pseudorandom generators called \emph{quantum pseudorandom generators (QPRGs)} can be based on the existence of logarithmic output length PRSGs. Our result along with the previous separations gives a better picture regarding the relationship between the two notions. We also study the relationship between other notions, namely, pseudorandom function-like state generators and pseudorandom functions. We provide evidence that QPRGs can be as useful as PRGs by providing cryptographic applications of QPRGs such as commitments and encryption schemes. 
\par Our primary technical contribution is a method for pseudodeterministically extracting uniformly random strings from Haar-random states.  

\end{abstract}


\newpage 
\tableofcontents
\newpage 

\input{intro}
\input{TechOverview}
\input{prelims}

\input{detextraction}

\input{qpseud}
\input{applications}

\printbibliography

\appendix 
\input{ProofAmplification}

\end{document}

%% file: headers.tex
\newcommand{\hybrid}{\mathsf{H}}

\newcommand{\ketbra}[2]{\ket{#1}\!\bra{#2}}
\renewcommand{\cal}[1]{\mathcal{#1}}
\newcommand{\R}{\mathbb{R}}
\newcommand{\C}{\mathbb{C}}
\newcommand{\N}{\mathbb{N}}
\newcommand{\E}{\mathop{\mathbb{E}}}
\newcommand{\Tr}{\mathrm{Tr}}

\newcommand{\eps}{\varepsilon}

\newcommand{\Haar}{\mathscr{H}}

\newcommand{\Enc}{\mathsf{Enc}}
\newcommand{\Dec}{\mathsf{Dec}}

\DeclareMathOperator{\TD}{TD}

\newcommand{\norm}[1]{{\left\lVert#1\right\rVert}}

\newcommand{\bit}{{\{0, 1\}}}

\newcommand{\cH}{\mathcal{H}}
\newcommand{\dmx}{\mathcal{D}}

\newcommand{\secparam}{\lambda}

\DeclareFontFamily{U}{skulls}{}
\DeclareFontShape{U}{skulls}{m}{n}{ <-> skull }{}

\newcommand{\prob}{{\sf Pr}}
\newcommand{\ext}{{\sf Ext}}
\newcommand{\prfs}{{\sf PRFS}}

\newcommand{\negl}{\mathsf{negl}}
\newcommand{\poly}{\mathsf{poly}}
\newcommand{\secp}{\secparam}

\newcommand{\Ex}{\E}
\newcommand{\argmax}{\operatornamewithlimits{argmax}}

\newcommand{\iid}{\text{i.i.d.} }

\DeclarePairedDelimiter\floor{\lfloor}{\rfloor}

\newcommand{\cD}{\mathcal{D}}

\newcommand{\cG}{\mathcal{G}}
\newcommand{\cK}{\mathcal{K}}
\newcommand{\cN}{\mathcal{N}}

\newcommand{\cS}{\mathcal{S}}
\newcommand{\cU}{\mathcal{U}}
\newcommand{\cX}{\mathcal{X}}

\newcommand{\round}{{\sf Round}}

\newcommand{\sfE}{{\sf E}}

\newcommand{\Good}{{\sf Good}}
\newcommand{\Bad}{{\sf Bad}}

\renewcommand{\Enc}{\mathsf{Enc}}
\renewcommand{\Dec}{\mathsf{Dec}}

\newcommand{\gen}{\mathsf{Gen}}
\newcommand{\ver}{\mathsf{Ver}}

\newcommand{\Com}{\mathsf{Com}}

\newcommand{\tomography}{\mathsf{Tomography}}
\newcommand{\prs}{\mathsf{PRS}}

\newcommand{\wqprg}{\mathsf{wQPRG}}
\newcommand{\sqprg}{\mathsf{sQPRG}}

\renewcommand{\Re}{{\sf Re}}
\renewcommand{\Im}{{\sf Im}}
\newcommand{\st}{\text{ s.t. }}
\newcommand{\by}{\times}
\newcommand{\veps}{\varepsilon}
\newcommand{\vtheta}{\vartheta}
\newcommand{\MAJ}{{\sf MAJ}}

\newcommand{\var}{\mathsf{Var}}
\newcommand{\xor}{\oplus}
\newcommand{\bigxor}{\bigoplus}
\renewcommand{\set}[1]{\left\{#1\right\}}

\newcommand{\expbraket}[3]{\langle#1|#2|#3\rangle}

%% file: intro.tex
\section{Introduction}
Deterministically generating long pseudorandom strings from a few random bits is a fundamental task in classical cryptography. Pseudorandom generators (PRGs) are a primitive that achieves this task and are ubiquitous throughout cryptography. Beyond cryptography, pseudorandom generators have found applications in complexity theory~\cite{RR94,LP20} and derandomization~\cite{NW94,impagliazzo1997p}. 
\par The concept of pseudorandomness has also been explored in other contexts. Of interest is the notion of pseudorandom quantum states, a popular pseudorandomness notion studied in the quantum setting.  Pseudorandom quantum states (PRS), introduced by Ji, Liu, and Song~\cite{JLS18}, are efficiently computable states that are computationally indistinguishable from Haar-random states. 
PRS have found numerous applications in other areas such as physics~\cite{bouland2019computational,BFGVZ22}, quantum machine learning~\cite{HBC+22}, and cryptography~\cite{AQY21,MY21}. 
\par While some recent works make progress towards understanding the feasibility of PRSGs (PRS generators), its relationship with PRGs\footnote{We are interested in PRGs that guarantee security against efficient quantum adversaries. Typically, such PRGs are referred to as post-quantum PRGs. For brevity, henceforth, by PRGs we will mean post-quantum PRGs.} and its implications to cryptography, there are still some important gaps that are yet to be filled. While the directions listed below might come across as seemingly unrelated, we will discuss how our work addresses all these three different directions. \\

\noindent \underline{\textsc{Direction 1. Separating PR{\bf S}Gs and PRGs}} We summarize the implication from PRGs to PR{\bf S}Gs\footnote{The ``S" is emphasized in this paragraph to highlight the difference between PR{\bf S}Gs and PRGs.} and back in the table below.~\cite{JLS18} showed that any quantum-query secure PRF (implied by PRGs) implies the existence of $\omega(\secparam)$-output length PR{\bf S}Gs, where $\secparam$ is the seed length. On the other hand, Kretschmer~\cite{Kretschmer21} showed a separation between $\omega(\secparam)$-length PR{\bf S}Gs and PRGs. In the $c \cdot \log(\secparam)$-regime, where $c \in \mathbb{R}$ and $c \ll 1$,~\cite{BrakerskiS20} showed that $c \cdot \log(\secparam)$-length PR{\bf S}Gs can be constructed unconditionally. Hence there is a trivial separation between $c \cdot \log(\secparam)$-length PR{\bf S}Gs, with $c \ll 1$, and PRGs since the latter requires computational assumptions. In the same work,~\cite{BrakerskiS20} showed that $c \cdot \log(\secparam)$-length PR{\bf S}Gs, when $c \gg 1$, can be constructed from PRGs. However, whether PR{\bf S}Gs with output length at least $\log(\secparam)$ imply PRGs or are separated from it is currently unknown. 

\begin{center}
\begin{tabular}{|c|c|c|}
\hline
{\bf Output Length of PRSG} & {\bf Implied by PRG?} & {\bf Implication to PRG?} \\ \hline
$\omega(\log(\secparam))$ & Yes~\cite{JLS18} & Black-box separation~\cite{Kretschmer21} \\ \hline
$c \cdot \log(\secparam))$, $c \geq 1$ & Yes~\cite{BrakerskiS20} & {\color{red} {\bf unknown} }  \\ \hline
$c \cdot \log(\secparam))$, $c \ll 1$ & N/A & Separation \\
\ & (Information-theoretic~\cite{BrakerskiS20}) & (trivial) \\ \hline
\end{tabular}
\end{center} 

\noindent Thus, the following question has been left open. 

\begin{quote}
{\em Does $\log(\secparam)$-length PRS imply PRGs or is it (black-box) separated from PRGs?}
\end{quote}

\noindent \underline{\textsc{Direction 2. Hybrid Cryptography.}} While the recent results demonstrate constructions of quantum cryptographic tasks such as commitments, zero-knowledge and secure computation from assumptions potentially weaker than one-way functions, the main drawback of these constructions is that they require the existence of quantum communication channels, an undesirable feature. Starting with the work of Gavinsky~\cite{Gavinsky12}, there has been an effort in building quantum cryptographic primitives using classical communication channels. We can thus characterize the class of cryptographic primitives into three categories: classical cryptography (that uses only classical resources), hybrid cryptography (uses quantum computing but classical communication channels) and  quantum cryptography (no restrictions). Hybrid cryptography has the advantage that the primitives in this category could be based on assumptions weaker than classical cryptography but on the other hand, has the advantage that we only need classical communication channels. Towards a deeper understanding of hybrid cryptography, the following is a pertinent question:

\begin{quote}
\begin{center}
{\em  Identify foundational primitives in hybrid cryptography and understand their relationship with classical and quantum cryptographic primitives.}
\end{center}
\end{quote} 

\noindent \underline{\textsc{Direction 3. Domain Extension, Generically.}} Given any pseudorandom generator of output length $m$, we know how to generically transform it into another secure pseudorandom generator of length $\ell$, for any $\ell > m$. On the other hand, we have limited results for pseudorandom quantum states. Recently,~\cite{GJMZ23} showed that multi-copy $\omega(\secparam)$-length PRS implies a single-copy PRS with large output length. However, we are not aware of any length extension transformation that preserves the number of copies. Investigating this question will help us understand the relationship between PRSs of different output lengths. This leads to the following question: 

\begin{quote}
\begin{center}
{\em Can we generically transform a multi-copy $n$-output PRS into a multi-copy $\ell$-output PRS, where $\ell \gg n$? Or is there a black-box separation?}
\end{center} 
\end{quote}

\paragraph{Our Work.} Towards simultaneously addressing all three directions above, we introduce the notion of \emph{quantum pseudorandom generators (QPRGs)}: which are like classical PRGs in that the input is a short classical string and the output is a longer classical string that is computationally indistinguishable from uniform, but (a) generation algorithm is a quantum algorithm, and (b) the mapping from seed to output only has to be \emph{pseudodeterministic} (i.e., for a fixed seed, the output is a fixed string with high probability). We first show that assumptions that are plausibly weaker than the existence of classical OWFs/PRGs can be used to build QPRGs: we show that QPRGs can be constructed from logarithmic-output PRSGs. In other words, we can generate pseudorandom strings using pseudorandom quantum states in a (pseudo-)deterministic fashion. We then present cryptographic applications of QPRGs and highlight some implications for the structure of classical versus quantum cryptography. 

The reader might wonder whether the notion of quantum generation of classical pseudorandomness is trivial. After all, since quantum computation is inherently probabilistic and can generate unlimited randomness starting from a fixed input, why would one need \emph{pseudo}randomness? However, for cryptographic applications having a source of randomness is not enough; it is important that some random-looking string can be \emph{deterministically generated} using a secret key.

\subsection{Our Results}

\paragraph{Quantum PRGs from PRS.} 


Informally, a $(1 - \eps)$-pseudodeterministic QPRG is a quantum algorithm $G$ where
\begin{itemize}
    \item (\emph{Pseudodeterminism}) For $1 - \eps$ fraction of seeds $k \in \{0,1\}^\secparam$ outputs a fixed string $y_k \in \{0,1\}^n$ with probability at least $1 - \eps$, and
    \item (\emph{Pseudorandomness}) For all efficient quantum distinguishers $A$, 
\[
    \Big | \Pr_{k \leftarrow \{0,1\}^\secparam} [ A(G(k)) = 1] - \Pr_{y \leftarrow \{0,1\}^n} [A(y) = 1] \Big | \leq \negl(\secparam)~.
\]
In other words, no efficient quantum adversary can distinguish between the output of the generator and a uniformly random string. 
\end{itemize}
(See \Cref{section:Quantum Pseudorandomess} for a formal definition of QPRGs). Our first result is the following:

\begin{theorem}[Informal]
\label{thm:main:intro}
Assuming the existence of logarithmic PRS, there exist $\Big(1 - \frac{1}{\poly(\secparam)} \Big)$-pseudodeterministic QPRGs.
\end{theorem}

\noindent Our result has several implications. 

\begin{itemize}
    \item \textsc{Implication to Direction 1 (Separating PRSGs and PRGs).} Perhaps surprisingly, our result suggests that in the logarithmic output regime, PRS and PRGs are not separated. This is unlike the super-logarithmic output regime and sub-logarithmic output regime both of which are separated from PRGs. In contrast, in the classical setting, achieving more pseudorandom bits is harder (in some cases strictly~\cite{ODW14}) than achieving a few pseudorandom bits. Our result when combined with prior results~\cite{JLS18,BrakerskiS20} gives a better picture on the relationship between PRS and PRGs (refer to the table in the introduction). 
    \item \textsc{Implication to Direction 2 (Hybrid Cryptography).} Quantum pseudorandom generator is a hybrid cryptographic primitive since it has a quantum generation algorithm but classical inputs and outputs. Later, we show that QPRGs imply many other hybrid cryptographic primitives such as quantum bit commitments with classical communication, quantum encryption with classical ciphertexts, and so on. Our results suggest that QPRGs could play a similar role in hybrid cryptography as PRGs did in classical cryptography. Proving whether QPRGs is a minimal assumption in hybrid cryptography, akin to how PRGs is a minimal assumption in classical cryptography~\cite{Gol90}, is an interesting open question. Furthermore, our result above highlights connections between hybrid and quantum cryptography. 
    
    \item \textsc{Implication to Direction 3 (Domain Extension, Generically).} In the above result, unfortunately, we only obtain QPRGs with inverse polynomial pseudodeterminism error. Suppose we can reduce the error to be negligible then we claim that $O(\log(\secparam))$-output PRS can be generically transformed into $\omega(\log(\secparam))$-output PRS. This can be achieved by appropriately instantiating the construction of Ji, Liu, and Song~\cite{JLS18} using quantum pseudorandom generators, which in turn can be built from $O(\log(\secparam))$-output PRS. Thus, the question of whether we can generically increase the output length of PRS is related to the question of reducing pseudodeterminism error in the above theorem. 

    \item \textsc{Other Implications: New Approach to Pseudorandomness.} One implication of our QPRG construction is that it demonstrates an ``inherently quantum'' way to generate classical pseudorandomness. There are plausible candidates for PRS (even the logarithmic-length ones) that don't seem to involve any classical OWFs in them at all; for example, it is conjectured that random polynomial-size quantum circuits generate pseudorandom states~\cite{AQY21}. 
     
\end{itemize}

\paragraph{Applications of Quantum PRGs.} Next we investigate the cryptographic applications of QPRGs. We demonstrate that QPRGs can effectively replace classical pseudorandom generators in some applications; although the QPRGs are not entirely deterministic, being $(1-\frac{1}{\poly(\secparam)})$-pseudodeterministic is good enough.

Concretely, we explore two applications: statistically binding and computationally hiding commitments, and pseudo one-time pads. While~\cite{AGQY22} previously demonstrated that these applications can be based on logarithmic PRS, we provide alternate proofs assuming the existence of QPRGs combined with Theorem~\ref{thm:main:intro}. Moreover, our constructions resemble the textbook constructions of classical commitments and pseudo one-time pads and thus, are conceptually simpler than the ones presented by~\cite{AGQY22}. 

A statistically binding commitment scheme is a fundamental cryptographic notion where a sender commits to a value such that it is infeasible, even if it is computationally unbounded, for them to change their commitment to a different value. Statistically binding quantum commitments have been a critical tool to achieve another fundamental notion in cryptography, namely secure computation~\cite{BartusekCKM21a,GLSV21}. We demonstrate that statistically binding and computationally hiding commitments can be constructed from QPRGs. 

\begin{theorem}[Informal]
\label{thm:intro:commitments}
Assuming the existence of $(1 - \frac{1}{\poly(\secparam)})$-pseudodetermininistic QPRGs, there exist statistically binding and computationally hiding quantum commitments with classical communication.
\end{theorem}

\noindent It is worth mentioning that there is another recent  work~\cite{BBSS23eprint} that also builds quantum commitments with classical communication albeit from incomparable assumptions\footnote{They consider a variant of PRS referred to as PRS with proof of deletion. On one hand, they don't have restriction on the output length like we do and on the other hand, they assume that PRS satisfies the additional proof of deletion property whereas we don't.}. 

Pseudo one-time pads are a variation of the one-time pad encryption scheme, where the encryption key is much smaller than the message length. As demonstrated by~\cite{AQY21}, pseudo one-time pads are useful for constructing classical garbling schemes~\cite{AIK06} and quantum garbling schemes~\cite{BY22}, which have numerous applications in cryptography. We demonstrate that pseudo one-time pads can be constructed from QPRGs.

\begin{theorem}[Informal]
Assuming the existence of $(1 - \frac{1}{\poly(\secparam)})$-pseudodetermininistic QPRGs, there exist pseudo one-time pads.
\end{theorem}

\paragraph{Quantum PRFs.} In addition to the above, we also explore pseudorandom functions with a quantum generation algorithm, which we call quantum pseudorandom functions (QPRFs). We show the following theorem:

\begin{theorem}[Informal]
Assuming the existence of $(\omega(\log\secp), O(\log\secp))$-PRFS, there exists a quantum pseudorandom function (QPRF) satisfying determinism with probability at least $\left( 1-\frac{1}{\poly(\secp)} \right)$.
\end{theorem}

In the above theorem, we use pseudorandom function-like states~\cite{AQY21}, a quantum analog of pseudorandom functions, to accomplish this. The notion of pseudorandom function-like states says the following: $t$ copies of states $(\ket{\psi_{x_1}},\ldots,\ket{\psi_{x_q}})$ are computationally indistinguishable from $t$ copies of $q$ Haar states, where $\ket{\psi_{x_i}}$, for every $i \in [q]$, is produced using an efficient PRFS generator that receives as input a key $k \in \set{0,1}^{\secparam}$, picked uniformly at random, and an input $x_i \in \set{0,1}^{\secparam}$. Just like in the case of QPRGs, in the above theorem, we  require PRFS with logarithmic input length.

We show how to leverage QPRFs to achieve private-key encryption with QPT algorithms and classical communication, which is the first result to achieve this notion from assumptions potentially weaker than one-way functions.

\subsection{Future Directions}
Our research raises several important open questions that remain to be explored. Below, we highlight two particularly interesting ones.

\paragraph{Separating QPRGs and QPRFs from Classical Cryptography.} While QPRGs and QPRFs are similar in flavor to their classical counterparts, their ability to generate quantum states suggests that they may be based on weaker assumptions than classical pseudorandom generators and functions. A key question is whether there is a fundamental separation between QPRGs and PRGs, as well as between QPRFs and PRFs. Proving that there is no separation would require a mechanism to efficiently dequantize the generation algorithm, which is a challenging task. This is especially true if the quantum generation algorithm involves running a quantum algorithm that is believed to be difficult to efficiently dequantize; for example, Shor's algorithm.

\paragraph{Reducing Determinism Error.} One  limitation of both QPRGs and QPRFs is that they suffer from inverse polynomial determinism error. It would be interesting to explore whether this error can be reduced to negligible, or whether a negative result can be proven. Understanding the fundamental limits of determinism in quantum pseudorandomness could have important implications. For instance, due to the inverse polynomial error, it is unclear how to apply the GGM transformation~\cite{GGM86} to go from quantum pseudorandom generators to quantum pseudorandom functions.

%% file: TechOverview.tex
\subsection{Technical Overview}
\label{sec:tech-overview}
We summarise our technical contributions below: 
\begin{itemize}
    \item We identify the definition of pseudodeterministic extractor that gives quantum pseudorandom generators. We then realize the notion of pseudodeterministic extractors; this is our core technical contribution and it involves using interesting properties about the Haar measure in the analysis. 
    \item We define and realize quantum pseudorandom functions from logarithmic output pseudorandom function-like states. Defining quantum PRFs turns out to be subtle. 
    \item We demonstrate applications of quantum pseudorandom generators and functions to commitments and encryption schemes. Especially, in commitment schemes, it turns out to be tricky to argue security due to the inverse polynomial determinism error. 
\end{itemize}

\subsubsection{Core Contribution: Pseudodeterministic Extractor}
\label{sec:corecontribution}

\noindent We focus on the goal of building a quantum pseudorandom generator from $O(\log(\secparam))$-qubit pseudorandom quantum states. Towards this goal, we identify the important step as follows: extracting $\poly(d(\secparam))$-length binary strings from $\log(d(\secparam))$-qubit Haar states in such a way, the following key properties are satisfied: 
\begin{itemize}
\item \underline{\textsc{Pseudodeterminism}}:  Running the extraction process on $\poly(d(\secparam))$-copies of $|\psi \rangle$ should give the same string $y$ with a very high probability. Ideally, with probability at least $1-\frac{1}{\poly(d(\secparam))}$, 
\item \underline{\textsc{Efficiency}}: The extraction process should run in time $\poly(d(\secparam))$,
\item \underline{\textsc{Statistical Indistinguishability}}:  The string $y$ is statistically close to the uniform distribution over $\{0,1\}^{\poly(d(\secparam))}$ as long as $| \psi \rangle$ is sampled from the Haar distribution. Here, we allow the total variation distance error to be as large as $O(\frac{1}{d})$.  
\end{itemize} 
\vspace{-1em}

\noindent It turns out most of the work goes in achieving the pseudodeterminism property. 

\paragraph{Toy Case.} Towards designing an extractor satisfying the above three properties, we first consider an alternate task. Instead of $\poly(d(\secparam))$-copies of the $\log(d(\secparam))$-qubit state $|\psi \rangle$, we are given all the amplitudes of $|\psi \rangle$, say $(\alpha_1,\ldots,\alpha_{d(\secparam)})$, in the clear. Can we extract true randomness from this? For instance, we could extract $b_1,\ldots,b_{d(\secparam)}$, where $b_i$ is the first bit of the real component of $\alpha_i$. Firstly, it is not even clear that $b_i$ is distributed to according to the uniform distribution over $\{0,1\}$. Moreover, all the bits $b_1,\ldots,b_{d(\secparam)}$ are not independent and in fact, are correlated with each other due to the normalization condition $\sum_i |\alpha_i|^2=1$. 
\par Fortunately, we can rely upon a result in random matrix theory~\cite{Meckes19}, that states the following: suppose $(\alpha_1,\ldots,\alpha_{d(\secparam)})$ are drawn from a Haar measure in ${\cal S}(\mathbb{R}^{d})$ then it holds that any $o(d(\secparam))$ co-ordinates of $(\alpha_1,\ldots,\alpha_{d(\secparam)})$ are $1/o(d(\secparam))$-close in total variation distance with $o(d(\secparam))$-dimensional vector where each component is drawn from i.i.d Gaussian ${\cal N}(0,\frac{1}{d})$. 
\par We generalize this result to the case when $(\alpha_1,\ldots,\alpha_{d(\secparam)})$ are drawn from a Haar measure in ${\cal S}(\mathbb{C}^d)$, and not ${\cal S}(\mathbb{R}^d)$ (see~\Cref{cor:Haar_Gaussian}) at the cost of reducing the standard deviation from $\frac{1}{d}$ to $\frac{1}{2d}$. We then use our observation to come up with an extractor as follows. The extractor takes as input\footnote{For the current discussion, we assume that the extractor has an infinite input tape that allows for storing infinite bits of precision of the complex numbers.} $(\alpha_1,\ldots,\alpha_{d(\secparam)})$,
\begin{itemize}
    \item Choose the first $k=o(d(\secparam))$ entries among $(\alpha_1,\ldots,\alpha_{d(\secparam)})$.
    \item Rounding step: for every $i \in [k]$, if $\Re(\alpha_i) > 0$, then set $b_i=0$. Otherwise, set $b_i=1$. 
    \item Output $b_1 \cdots b_{k}$. 
\end{itemize}
From our observation and the symmetricity of ${\cal N}(0,\frac{1}{d})$, it follows that when $(\alpha_1,\ldots,\alpha_{d(\secparam)})$ is drawn from a Haar distribution on ${\cal S}(\mathbb{C}^{d})$ then the output of the extractor is $o(\frac{1}{d(\secparam)})$-close to uniform distribution on $\{0,1\}^k$. Moreover, the above procedure is deterministic. 

\paragraph{Challenges.} Our hope is to leverage the above ideas to design an extractor that can extract given $\poly(d(\secparam))$-copies of a $O(\log(d(\secparam)))$-qubit Haar state $\ket{\psi}$. We encounter a couple of challenges. 
\begin{enumerate}
    \item First challenge: We have access only to the copies of $\ketbra{\psi}{\psi}$ without the amplitudes given to us in plain text, making it infeasible to implement the previously described method. However, we can still carry out tomography and retrieve an estimated version of the matrix $\ketbra{\psi}{\psi}$. If the amplitudes of $\ket{\psi}$ are $\{\alpha_x\}_{x \in [d]}$ then the $(x,y)^{th}$ entry in the density matrix $\ketbra{\psi}{\psi}$ is $\alpha_x \alpha_y^{*}$. We need to analyze the distribution corresponding to $\alpha_x \alpha_y^{*}$ and, design an approach for obtaining a uniform distribution from it.
    \item Second challenge: Tomography is inherently a probabilistic technique, and hence, each time tomography is executed on multiple copies of $\ket{\psi}$, the output obtained may differ. Additionally, the trace distance between the density matrix obtained via tomography and the original density matrix is inversely proportional to the dimension, which is polynomial in this case, and this may be significant. Both of these factors collectively affect the determinism guarantees of the extractor. In general, it is not feasible to partition ${\cal S}(\C^d)$ into regions labeled by a bitstring such that given multiple copies of a state in a region, the corresponding bitstring can be deterministically recovered.

\end{enumerate}
\noindent We tackle the above challenges using the following insights.

\paragraph{Addressing the first challenge.} We first tackle the first bullet above. Notice that the diagonal entries in the density matrix $\ketbra{\psi}{\psi}$ is $\{|\alpha_i|^2\}_{i \in [d(\secparam)]}$, where $\ket{\psi}=\sum_i \alpha_i \ket{i}$. If $\alpha_j = a_j + i b_j$ then $|\alpha_j|^2 = a_j^2 + b_j^2$. Given our earlier observation about the closeness of $o(d(\secparam))$ entries in a vector drawn from ${\cal S}(\mathbb{C}^d)$ with iid Gaussian, we will make the following simplifying assumption. We assume that   $(\alpha_1,\ldots,\alpha_{k})$, where $k=o(d)$, is sampled such that for every $i \in [k]$, $a_i$ and $b_i$ are distributed according to i.i.d Gaussian ${\cal N}(0,\frac{1}{2d})$. From this, it follows that $|\alpha_i|^2$ is distributed according to a {\em chi-squared} distribution with 2 degrees of freedom. Unfortunately, chi-squared distribution does not have the same nice symmetricity property as a Gaussian distribution. So we will instead extract randomness in a different way. 
\par We divide $(|\alpha_1|^2,\ldots,|\alpha_k|^2)$ into blocks of size $r$ and denote $\ell$ to be the number of blocks, where $r,\ell = o(d)$. Then, add the elements in a block. Call the resulting elements $q_1,\ldots,q_{\ell}$. From central limit theorems~\cite{sirazhdinov1962convergence}, one can show that $q_1,\ldots,q_{\ell}$ are $O(1/\sqrt{r})$-close to $\ell$ samples drawn i.i.d from $\mathcal{N}(\frac{r}{d},\frac{r}{d^2})$. Thus, using central limit theorem, we are back to the normal distribution, except that the mean is shifted to $\frac{r}{d}$ rather than 0. This gives rise to a natural rounding mechanism. 
\par We will check if $q_i > \frac{r}{d}$ and if so, we set a bit $b_i=0$ and if not, we set it to be 0. By carefully choosing the parameters $k$ and $\ell$ and combining the above observations, we can argue that $b_1,\ldots,b_{k}$ is $O(d^{-1/6})$-close to the uniform distribution on $\{0,1\}^{\ell}$. 
\par To summarise, the informal description of the extractor is as follows: given $\poly(d)$ copies of a $d$-dimensional state $\ket{\psi}$,
\begin{itemize}
\item First perform tomography to recover a matrix $M \in \mathbb{C}^d \times \mathbb{C}^d$ that is an approximation of $\ketbra{\psi}{\psi}$
\item Then, pick $o(d)$ diagonal entries in $M$ and break this into $\ell$ blocks of size $r$. 
\item Sum up all the entries in each block to get $\ell$ values $q_1,\ldots,q_{\ell}$. Round every $q_i$ to get $b_i$. 
\item Output $b_1,\ldots,b_{\ell}$.
\end{itemize}

\paragraph{Addressing the second challenge.} While the above construction seems promising, we still have not addressed the second challenge pertaining to the determinism property. It could be the case that all the $q_i$s are very close to the mean and due to the tomography error, every time we try to extract we set $b_i=0$ sometimes and $b_i=1$ the rest of the time. This should not be surprising as we said earlier, that it should not be possible to partition ${\cal S}(\mathbb{C}^d)$ such that for every $\ket{\psi}$, there is a bitstring $b_{\psi}$ such that given many copies of $\ket{\psi}$, the extractor always outputs the same bitstring $b_{\psi}$.   
\par In fact, we can identify a {\em forbidden region} in ${\cal N}(\frac{r}{d},\frac{r}{d^2})$ (see~\Cref{fig:gaussian} below) such that if $q_i$ falls into the forbidden region then there is a significant chance that $q_i$ will be classified as either 0 or 1. The forbidden region has width $\frac{1}{d}$ on either side of the mean. Given this, we give up all hope of achieving perfect determinism and instead shoot for determinism with $o(1/d)$ error. 
\par We identify a set of $d(\secparam)$-dimensional states ${\cal G}_{\Delta}$, where $\Delta=\frac{1}{d}$, such that if a state $\ket{\psi}$ is in ${\cal G}_{\Delta}$ then it holds that none of $q_1,\ldots,q_{\ell}$, generated from $\ket{\psi}$, lies in the forbidden region. The setting of $\Delta$ is carefully chosen to accommodate for the error in tomography. 

\begin{figure}[!htb]
\begin{center}
\begin{tikzpicture}[>={Stealth[length=6pt]},declare function={g(\x)=2*exp(-\x*\x/3);
    xmax=3.5;xmin=-3.4;x0=1.5;ymax=2.75;}]
    \path (4,0) node[below]{$x$} (x0,0) (0,ymax) node[right]{${\cal N}\left( \frac{r}{d},\frac{r}{d^2} \right)$};
 \draw[black!50] (-3.7,0) edge[->] (4,0) foreach \X in {-3.5,-3,...,3}
  {(\X,0) -- ++ (0,0.1)} (0,0) edge[->] (0,ymax);
 \draw[fill=red!1000,opacity=0.1] plot[domain=-0.75:0.75,samples=15,smooth] (\x,{g(\x)}) -- (0.75,0) -| cycle;  
 \draw[>=stealth,line width=0.01mm,<->] (0,-0.1) -- node[below]{$\frac{1}{d}$} (0.75,-0.1);
 \draw plot[domain=xmin:xmax,samples=51,smooth] (\x,{g(\x)}); 
\end{tikzpicture}
\end{center}
\vspace{-2em}
\caption{The red region denotes the {\em forbidden region}.}
\label{fig:gaussian}
\end{figure}
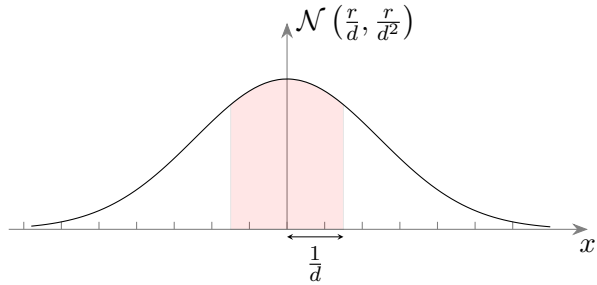

\noindent Once ${\cal G}_{\Delta}$ is identified, we prove two things: 
\begin{itemize}
    \item Firstly, if a state is sampled from the Haar distribution on ${\cal S}(\mathbb{C}^{d(\secparam)})$ then with at least $1-o\left( \frac{1}{d} \right)$ probability, $\ket{\psi} \in {\cal G}_{\Delta}$.  
    \item Secondly, for every $\ket{\psi} \in \mathcal{G}_{\Delta}$, the probability that the extractor, given $\poly(d(\secparam))$-copies of $\ket{\psi}$, outputs the same string twice is at least $1-o\left( \frac{1}{d} \right)$. Roughly, this follows from the fact that $(q_1,\ldots,q_{\ell})$, generated from $\ket{\psi}$, gets misclassified with very small probability. 
\end{itemize}
We can leverage the above two observations to show that our extractor satisfies determinism with probability at least $1-o\left( \frac{1}{d} \right)$. 

\subsubsection{From Pseudodeterminism Extractor to Quantum PRGs} 
With the pseudodeterministic extractor in hand, we propose the following construction of quantum pseudorandom generators: on input a seed $k \in \{0,1\}^{\secparam}$,
\begin{itemize}
    \item Perform this procedure polynomially many times: run the PRS generator on $k$ to produce the PRS state $\ket{\psi}$, 
    \item Run the pseudodeterministic extractor (discussed in~\Cref{sec:corecontribution}) on polynomially many copies of $\ket{\psi}$ to obtain a binary string $y$,
    \item Output $y$. 
\end{itemize}
\noindent Recall that the guarantees of the pseudodeterministic extractor only hold for Haar states. In the above construction, we are invoking the extractor on PRS states. However, we can invoke the security of PRSGs to replace PRS states with Haar states and then invoke the guarantees of the extractor. Just like the extractor, the above quantum PRG would also suffer from an inverse polynomial determinism error. Similarly, the above construction would suffer from inverse polynomial error in security; that is, the output of the above QPRG (on a random seed) cannot be distinguished from a uniformly random output with at most inverse polynomial error. We call a QPRG that satisfies inverse polynomial determinism error and inverse polynomial security error to be a {\em weak} QPRG and a QPRG that satisfies negligible security error to be a {\em strong} QPRG. 
\par While we currently do not know how to reduce the pseudodeterminism error, there is still hope to reduce the security error. Indeed, there are security amplification techniques that are well studied in the classical cryptography literature. 

\paragraph{From Weak QPRG to Strong QPRG.} The naive approach of going from a weak QPRG to a strong QPRG is to use parallel repetition: on input a seed of length $s \cdot \secparam$, break the seed into $s$ parts, apply QPRG on each of them and then XOR the outputs. While this should help security (as we see below), it hurts pseudodeterminism. Using union bound, we can argue that the pseudodeterminism error increases by a multiplicative factor of $\secparam$. Thus, if we start with a weak PRG with a sufficiently small pseudodeterminism error then this multiplicative factor won't hurt us. 

To prove that the security error is negligible, we will use XOR-based security amplification techniques~\cite{DIJK09TCC, MauTes09CRYPTO, MauTes10TCC}. The analysis in the amplification theorems~\cite{DIJK09TCC, MauTes09CRYPTO, MauTes10TCC} was initially tailored to classical settings, using a careful analysis, we show that the analysis also extends to the quantum setting. 

\subsubsection{Quantum Pseudorandom Functions} 
We also demonstrate the connections between quantum pseudorandom functions and logarithmic-output pseudorandom function-like states~\cite{AGQY22}. 
\par Roughly speaking, a quantum pseudorandom function is a pseudorandom function except that the generation algorithm is quantum and moreover, we allow for inverse polynomial determinism error. Defining the security of this notion requires some care. If we allow the adversary to make arbitrary queries to the oracle (that is either the QPRF or the random function) then such a notion is clearly impossible to achieve. For instance, the adversary can query the same input twice; in the case when the oracle implements a random function, we always get the same output, and in the QPRF case, there is an inverse polynomial probability with which we get different outputs. Hence, we restrict our attention to the setting when the adversary only makes selective and distinct queries. We show that this definition is sufficient for applications. 
\par The construction of QPRFs from logarithmic-output PRFS follows along similar lines as the construction of QPRGs from pseudodeterminism extractors.

\subsubsection{Applications} 
\noindent We show that QPRGs imply both statistically binding commitments and pseudo-one-time pads. The constructions are similar to the existing constructions from (classical) pseudorandom generators except that we need to contend with the inverse polynomial determinism error. For most of the applications, naive parallel repetition and a majority argument are sufficient to circumvent the determinism issue. However, for the application of commitments, the analysis turns out to be relatively more complicated. 

\paragraph{Commitments.} 
\par The construction of statistically binding commitments from QPRGs is inspired by Naor commitments~\cite{Naor91}.
\par To recall Naor's construction: the receiver sends a random string $r$ of length $3\secparam$ to the sender who applies a classical pseudorandom generator with output length $3\secparam$ on a random seed of length $\secparam$. Depending on the message bit, the sender either XORs the output with $r$ or sends the PRG output as-is. The proof of binding relies upon the fact that the number of pairs of keys whose outputs when XORed with each other lead to $r$ is precisely upper bounded by $2^{2\secparam}$, which is negligible in comparison with all possible values of $r$. 
\par A natural modification to the above construction is to replace the PRG with quantum pseudorandom generator. The immediate issue that arises here is correctness due to the inverse polynomial determinism error. Again, using naive parallel repetition we can resolve the determinism error: where the sender computes many QPRG outputs on independent seeds and depending on the message, the outputs are either XORed with $r$ or kept as-is. In the modified construction, arguing hiding is fairly straightfoward. However, arguing the binding property requires some care. 
\par Naor's binding argument cannot be immediately generalized to the QPRG setting since it has inverse polynomial determinism error. However, we come up with a different argument in this setting: in two technical claims,~\Cref{claim:prob:inner} and~\Cref{claim:collision}, we prove that the statistical binding property still holds. Roughly speaking, the intuition behind the argument is as follows. Suppose ${\sf Bad}$ be the set of QPRG seeds where the pseuodeterminism error is too high; larger than any inverse polynomial and ${\sf Good}$ be the set containing the rest of the QPRG seeds. If the adversarial sender chooses from ${\sf Bad}$ in the commit phase (or even in the opening phase), it could only hurt itself because it will not be able to control the output of the QPRG during the verification process executed by the receiver in the opening phase. On the other hand, if the adversarial sender commits to seed from ${\sf Good}$ in the commit phase and sends (a possibly different) seed from ${\sf Good}$ in the opening phase then using the fact that the outputs are mostly deterministic, we can argue that with overwhelming probability over $r$, the XOR of the two seeds does not equal $r$.  

%% file: prelims.tex
\section{Preliminaries}
\label{sec:prelims}

We refer the reader to~\cite{nielsen_chuang_2010} for a comprehensive reference on the basics of quantum information and quantum computation. We use $I$ to denote the identity operator. We use $\cS(\cH)$ to denote the set of unit vectors in the Hilbert space $\cH$. We use $\dmx(\cal{H})$ to denote the set of density matrices in the Hilbert space $\cal{H}$. Let $P,Q$ be distributions. We use $d_{TV}(P,Q)$ to denote the total variation distance between them. Let $\rho,\sigma \in \dmx(\cal{H})$ be density matrices. We write $\TD(\rho,\sigma)$ to denote the trace distance between them, i.e.,
\[
    \TD(\rho,\sigma) = \frac{1}{2} \| \rho - \sigma \|_1
\]
where $\norm{X}_1 = \Tr(\sqrt{X^\dagger X})$ denotes the trace norm.
We denote $\norm{X} := \sup_{\ket\psi}\{\braket{\psi|X|\psi}\}$ to be the operator norm where the supremum is taken over all unit vectors.
For a vector $\ket{x}$, we denote its Euclidean norm to be $\norm{\ket{x}}_2$.
We use the notation $M\ge 0$ to denote the fact that $M$ is positive semi-definite.

\paragraph{Haar Measure.} The Haar measure over $\C^d$, denoted by $\Haar(\C^d)$ is the uniform measure over all $d$-dimensional unit vectors. One useful property of the Haar measure is that for all $d$-dimensional unitary matrices $U$, if a random vector $\ket{\psi}$ is distributed according to the Haar measure $\Haar(\C^d)$, then the state $U\ket{\psi}$ is also distributed according to the Haar measure. For notational convenience we write $\Haar_m$ to denote the Haar measure over $m$-qubit space, or $\Haar((\C^2)^{\otimes m})$.

\subsection{Quantum Algorithms}
\label{sec:algorithms}

A quantum algorithm $A$ is a family of generalized quantum circuits $\{A_\lambda\}_{\lambda \in \N}$ over a discrete universal gate set (such as $\{ CNOT, H, T \}$). By generalized, we mean that such circuits can have a subset of input qubits that are designated to be initialized in the zero state, and a subset of output qubits that are designated to be traced out at the end of the computation. Thus a generalized quantum circuit $A_\lambda$ corresponds to a \emph{quantum channel}, which is a is a completely positive trace-preserving (CPTP) map. When we write $A_\lambda(\rho)$ for some density matrix $\rho$, we mean the output of the generalized circuit $A_\lambda$ on input $\rho$. If we only take the quantum gates of $A_\lambda$ and ignore the subset of input/output qubits that are initialized to zeroes/traced out, then we get the \emph{unitary part} of $A_\lambda$, which corresponds to a unitary operator which we denote by $\hat{A}_\lambda$. The \emph{size} of a generalized quantum circuit is the number of gates in it, plus the number of input and output qubits.

We say that $A = \{A_\lambda\}_\lambda$ is a \emph{quantum polynomial-time (QPT) algorithm} if there exists a polynomial $p$ such that the size of each circuit $A_\lambda$ is at most $p(\lambda)$. We furthermore say that $A$ is \emph{uniform} if there exists a deterministic polynomial-time Turing machine $M$ that on input $1^\lambda$ outputs the description of $A_\lambda$. 

We also define the notion of a \emph{non-uniform} QPT algorithm $A$ that consists of a family $\{(A_\lambda,\rho_\lambda) \}_\lambda$ where $\{A_\lambda\}_\lambda$ is a polynomial-size family of circuits (not necessarily uniformly generated), and for each $\lambda$ there is additionally a subset of input qubits of $A_\lambda$ that are designated to be initialized with the density matrix $\rho_\lambda$ of polynomial length. This is intended to model nonuniform quantum adversaries who may receive quantum states as advice.
Nevertheless, the reductions we show in this work are all uniform.

The notation we use to describe the inputs/outputs of quantum algorithms will largely mimick what is used in the classical cryptography literature. For example, for a state generator algorithm $G$, we write $G_\lambda(k)$ to denote running the generalized quantum circuit $G_\lambda$ on input $\ketbra{k}{k}$, which outputs a state $\rho_k$.

Ultimately, all inputs to a quantum circuit are density matrices. However, we mix-and-match between classical, pure state, and density matrix notation; for example, we may write $A_\lambda(k,\ket{\theta},\rho)$ to denote running the circuit $A_\lambda$ on input $\ketbra{k}{k} \otimes \ketbra{\theta}{\theta} \otimes \rho$. In general, we will not explain all the input and output sizes of every quantum circuit in excruciating detail; we will implicitly assume that a quantum circuit in question has the appropriate number of input and output qubits as required by context. 

\subsection{Pseudorandomness Notions}
The notion of pseudorandom quantum states was first introduced by Ji, Liu, and Song in~\cite{JLS18}. We present the following relaxed definition of pseudorandom state (PRS) generators.\footnote{In~\cite{JLS18}, the output of the generator needs to be pure; while we allow it to be mixed.} We note that the relaxation is due to~\cite{AQY21}.

\begin{definition}[Pseudorandom State (PRS) Generator]
We say that a QPT algorithm $G$ is a \emph{pseudorandom state (PRS) generator} if the following holds.
\begin{enumerate}
    \item {\bf{State Generation.}} For all $\secp\in\N$ and all $k\in\bit^\secp$, the algorithm $G$ behaves as $G_\secp(k) = \rho_k$ for some $n(\secp)$-qubit (possibly mixed) quantum state $\rho_k$.
    \item {\bf{Pseudorandomness.}} For all polynomials $t(\cdot)$ and any (non-uniform) QPT distinguisher $A$, there exists a negligible function $\veps(\cdot)$ such that for all $\secp\in\N$, we have
    \begin{align*}
        \left| \Pr_{k\gets\bit^\secp}[A_\secp(G_\secp(k)^{\otimes t(\secp)}) = 1]
        - \Pr_{\ket{\vtheta}\gets\Haar_{n(\secp)}}[A_\secp(\ket{\vtheta}^{\otimes t(\secp)}) = 1] \right|
        \le \veps(\secp).
    \end{align*}
\end{enumerate}
\noindent We also say that $G$ is an \emph{$n(\secp)$-PRS generator} to succinctly indicate that the output length of $G$ is $n(\secp)$.
\end{definition}

\begin{definition}[Selectively Secure Pseudorandom Function-Like State (PRFS) Generators]
We say that a QPT algorithm $G$ is a \emph{selectively secure pseudorandom function-like state (PRFS) generater} if for all polynomials $q(\cdot), t(\cdot)$, any (non-uniform) QPT distinguisher $A$, and any family of pairwise distinct indices 
$\left( \set{x_1, \dots, x_{q(\secp)}} \subseteq \bit^{m(\secp)}\} \right)_\secp$, there exists a negligible function $\eps(\cdot)$ such that for all $\secp\in\N$,
\begin{multline*}
\Big\lvert\Pr_{k\gets\bit^\secp}\left[A_\secp(x_1,\dots,x_{q(\secp)},G_\secp(k,x_1)^{\otimes t(\secp)},\dots,G_\secp(k,x_{q(\secp)})^{\otimes t(\secp)})=1\right] \\
- \Pr_{\ket{\vtheta_1},\dots,\ket{\vtheta_{q(\secp)}}\gets\Haar_{n(\secp)}}\left[A_\secp(x_1,\dots,x_{q(\secp)},\ket{\vtheta_1}^{\otimes t(\secp)},\dots,\ket{\vtheta_{q(\secp)}}^{\otimes t(\secp)})=1\right]\Big\rvert
\le \eps(\secp).
\end{multline*}
We also say that $G$ is an \emph{$(m(\secp),n(\secp))$-PRFS generator} to succinctly indicate that its input length is $m(\secp)$ and its output length is $n(\secp)$.
\end{definition}

\subsection{Basics of Statistics and Haar Measure}
\noindent A simple yet useful observation is that for any two density matrices, the difference between any of their diagonal entries is bounded above by their trace distance.
\begin{fact}
\label{fact:TDdiagonal}
For any density matrices $\rho, \sigma \in \cD(\C^d)$, it holds that $\max_{i\in[d]}|\rho_{ii} - \sigma_{ii}| \le \TD(\rho, \sigma)$, where $\rho_{ii}, \sigma_{ii}$ denote the $i$-th diagonal entry of $\rho, \sigma$ respectively, i.e., $\rho_{ii} = \expbraket{i}{\rho}{i}$ and $\sigma_{ii} = \expbraket{i}{\sigma}{i}$.
\end{fact}
\begin{proof}
Note that the trace distance has the following variational form: 
\[
\TD(\rho, \sigma) = \max_{0\le M\le I} \Tr(M(\rho-\sigma)).
\]
Furthermore, trace distance is symmetric. Therefore, taking $M:=\ketbra{i}{i}$ for $i\in[d]$, we have $\TD(\rho, \sigma)\ge \max(\rho_{ii} - \sigma_{ii}, \sigma_{ii} - \rho_{ii}) = |\rho_{ii} - \sigma_{ii}|$ as desired.
\end{proof}

\begin{fact}
\label{fact:dataProcessingIneq}
Let $X, Y$ be random variables and $f$ be a function. Then $d_{TV}(f(X), f(Y))\le d_{TV}(X, Y)$.
\end{fact}

\begin{lemma}[Chernoff-Hoeffding Inequality]
\label{lemma:Chernoff-Hoeffding}
Let $X_1, X_2, \ldots, X_n$ be independent random variables, such that $0 \leq X_i \leq 1$ for all $i\in[n]$. Let $X = \sum_{i=1}^n X_i$ and $\mu = \Ex[X]$. Then for any $\veps > 0$,
\[
\Pr[ |X - \mu| > \veps)] \leq 2 e^{-\frac{2\veps^2}{n}}.
\]
\end{lemma}

\subsubsection{Chi-Squared Distributions}
We present the definition and properties of the chi-squared distribution in the following.

\begin{definition}[Chi-Squared Distribution]
Let $Z_1,\dots, Z_k$ be \iid Gaussian random variables $\cN(0, 1)$. The random variable 
\[
Q := \sum_{i\in[k]} Z_i^2.
\]
is distributed according to the \emph{chi-squared distribution with $k$ degrees of freedom}, denoted by $Q \sim \chi^2_k$.
\end{definition}

\begin{fact}
Let $Z\sim\cN(0,1)$. $Z^2$ has a finite third moment.
\end{fact}

\begin{fact}
\label{fact:chi-squared}
For all $k\in\N$, the following holds. Let $Q \sim \chi^2_k$. The mean of $Q$ is $k$ and the variance of $Q$ is $2k$. Moreover, suppose $Q_1\sim\chi^2_{k_1}$ and $Q_2\sim\chi^2_{k_2}$, then $Q_1 + Q_2\sim\chi^2_{k_1 + k_2}$. When $k=1$, we often omit the subscript and denote it by $\chi^2$.
\end{fact}

\noindent We introduce a strong version of the central limit theorem that characterizes the \emph{total variation distance} between the sum of \iid \emph{absolutely continuous}\footnote{A random variable $X$ is absolutely continuous if there exists a (probability density) function $f:\R\to[0,1]$ such that $\Pr[X\le x] = \int_{-\infty}^{x} f(t) \,dt$ for all $x\in\R$ and $\int_{-\infty}^{\infty} f(t) \,dt = 1$.} random variables and Gaussian random variables. Note that most versions of central limit theorems state only the convergence in \emph{cumulative density function}, which is not sufficient for our purpose.
\begin{lemma}[{\cite[Theorem 1]{sirazhdinov1962convergence}}, restated]
\label{lem:CLT}
Let $X_1,\dots, X_k$ be \iid random variables. If $X_1$ is absolutely continuous and has a finite third moment, then 
\[
d_{TV}\left( \frac{\sum_{i\in[k]}(X_i-\mu)}{\sqrt{k}\sigma}, Z \right) = O\left(\frac{1}{\sqrt{k}}\right),
\]
where $\mu$ is the mean of $X_1$, $\sigma$ is the standard deviation of $X_1$ and $Z \sim \cN(0,1)$.
Equivalently, $d_{TV}(\sum_{i\in[k]}X_i, Z') = O(1/\sqrt{k})$, where $Z' \sim \cN(k\mu, k\sigma^2)$.
\end{lemma}

\noindent Since a random variable with a chi-squared distribution is the sum of squared \iid Gaussian random variables, we have the following immediate corollary.

\begin{corollary}
\label{cor:CLTofChi}
Let $Q$ be a random variable with a distribution $\chi^2_k$. Then $d_{TV}(Q, Z) = O(1/\sqrt{k})$, where $Z \sim \cN(k, 2k)$.
\end{corollary}
\begin{proof}
By definition, $Q = \sum_{i\in[k]}Z_i^2$ where $Z_i\sim\cN(0, 1)$. It immediately follows from the facts that $Z_i^2$ is absolutely continuous, $\Ex\left[Z_i^2\right] = 1$, $\var\left(Z_i^2\right) = 2$, the third moment of $Z_i^2$ is finite and~\Cref{lem:CLT}.
\end{proof}

\subsubsection{Haar Measure}
\noindent Given a $d$-dimensional Haar state, all coordinates of the state are correlated due to the unit-norm condition. The following theorem states that the \emph{joint distribution} of $k = o(d)$ fraction of the coordinates in $\cS(R^d)$ is statistically close to a random vector with \iid Gaussian entries. The theorem was first proven in~\cite{DD87Haar}. We will use the version stated in~\cite{Meckes19}.

\begin{theorem}[{\cite[Theorem 2.8]{Meckes19}}]
\label{lem:dist:Haar_entry}
For every integer $d \ge 5$ and every $k\in\N$ that satisfies $1\le k\le d-4$, let $X = (X_1,\dots,X_d)$ be a uniform point on $\cS(\R^d)$. Let $Z$ be a random vector in $\R^k$ with \iid Gaussian entries $\mathcal{N}(0,1/d)$. Then
\begin{align*}
d_{TV} \left( (X_1,\dots,X_k), Z \right) \le \frac{2(k+2)}{d-k-3}.
\end{align*}
\end{theorem}

\noindent The above lemma can be extended to uniformly random vectors on $\cS(\C^d)$. For a complex number $\alpha$, we denote by $\Re(\alpha)$ and $\Im(\alpha)$, in order, the real part and imaginary part of $\alpha$.
\begin{lemma}
\label{fact:spherical:complex2real}
Let $\ket{\psi} = \sum_{i\in[d]}\alpha_i\ket{i}$ be a uniform point on $\cS(\C^d)$. Then $(\Re(\alpha_1),\dots,\Re(\alpha_d),$ $\Im(\alpha_1),\dots,\Im(\alpha_d))$ is a uniform point on $\cS(\R^{2d})$.
\end{lemma}
\begin{proof}
First proposed by~Muller~\cite{Mul59}, a uniform point on $\cS(\R^{2d})$ can be sampled via the following procedures:
\begin{enumerate}
    \item For $i\in[2d]$, sample $a_i\gets\cN(0,\sigma^2)$.
    \item Output $\sum_{i\in[2d]} \frac{a_i}{\sqrt{\sum_{j\in[2d]} a_j^2 }} \ket{i}$.
\end{enumerate}
Where $\sigma^2$ in step~1 could be an arbitrary positive number. 

\noindent On the other hand, a uniform point on $\cS(\C^{d})$ can be sampled as follows:
\begin{enumerate}
\item For $i\in[d]$, sample $\alpha_i \sim \C\cN(0,1)$.
\item Output $\sum_{i\in[d]} \frac{\alpha_i}{\sqrt{\sum_{j\in[d]} |\alpha_j|^2}} \ket{i}$.
\end{enumerate}
Particularly, in step~1, sampling $\alpha \sim \C\cN(0,1)$ is equivalent to sampling $\alpha = a + ib$ according to $a\sim\cN(0,1/2)$, $b\sim\cN(0,1/2)$ by the definition of the complex normal distribution. Hence, picking $\sigma^2 = 1/2$ completes the proof.
\end{proof}

\begin{corollary}
\label{cor:Haar_Gaussian}
For every integer $d \ge 3$ and every $k\in\N$ that satisfies $1\le 2k\le 2d-4$, let $\ket{\psi} = \sum_{i\in[d]}\alpha_i\ket{i}$ be a random point on $\cS(\C^d)$. Let $Z$ be a random vector in $\R^{2k}$ with \iid Gaussian entries $\cN(0,1/(2d))$. Then
\begin{align*}
    d_{TV}\left( (\Re(\alpha_1),\dots,\Re(\alpha_k), \Im(\alpha_1),\dots,\Im(\alpha_k)), Z \right) \le \frac{2(2k+2)}{2d-2k-3}.
\end{align*}
\end{corollary}
\begin{proof}
It immediately follows from \Cref{lem:dist:Haar_entry} and \Cref{fact:spherical:complex2real}.
\end{proof}

\noindent Next, the following simple fact gives an upper bound of the probability that a Gaussian random variable takes values near its mean.
\begin{fact}
\label{fact:NormalDist}
Let $Z\sim\cN(\mu, \sigma^2)$. For any $\Delta>0$,
\begin{align*}
    \Pr[|Z - \mu| \le \Delta] \le \sqrt{\frac{2}{\pi}} \frac{\Delta}{\sigma}.
\end{align*}
\end{fact}
\begin{proof}
Let $f(x) = \frac{1}{\sigma\sqrt{2\pi}}e^{ -\frac{1}{2}\left(\frac{x-\mu}{\sigma}\right)^2 }$ be the probability density function of $\cN(\mu, \sigma^2)$. 
The probability $\int_{\mu-\Delta}^{\mu+\Delta} f(x) \,dx$ can be upper-bounded by $f(\mu) \cdot 2\Delta = \sqrt{\frac{2}{\pi}}\frac{\Delta}{\sigma}$.
\end{proof}

\subsection{Quantum State Tomography}
\begin{lemma}[{\cite[Corollary 7.6]{AGQY22}}]
There exists a tomography procedure $\tomography$ that satisfies the following. For any error tolerance $\delta = \delta(\secp) \in (0,1]$ and any dimension $d = d(\secp) \in\N$, given at least $t = t(\secp) := 36\secp d^3/\delta$ copies of a $d$-dimensional density matrix $\rho$, $\tomography(\rho^{\otimes t})$ outputs a matrix $M\in\C^{d\by d}$ such that the following holds:
\begin{align*}
    \Pr\left[\| M - \rho\|^2_F \le \delta: M \gets \tomography(\rho^{\otimes t}) \right] 
    \ge 1 - \negl(\secp),
\end{align*}
where $\norm{\cdot}_F$ denotes the Frobenius norm. Moreover, the running time of $\tomography$ is polynomial in $1/\delta, d$ and $\secp$.
\end{lemma}

\noindent By using the fact that $\norm{A}_1 \le \sqrt{d}\norm{A}_F$, we have the following immediate corollary.
\begin{corollary}
\label{corollary:tomography}
There exists a tomography procedure $\tomography$ that satisfies the following. For any error tolerance $\delta = \delta(\secp) \in (0,1]$ and any dimension $d = d(\secp) \in\N$, given at least $t = t(\secp) := 144\secp d^4/\delta^2$ copies of a $d$-dimensional density matrix $\rho$, $\tomography(\rho^{\otimes t})$ outputs a matrix $M\in\C^{d\by d}$ such that the following holds:
\begin{align*}
    \Pr\left[ \TD(M, \rho) \le \delta: M \gets \tomography(\rho^{\otimes t})\right] 
    \ge 1 - \negl(\secp).
\end{align*}
Moreover, the running time of $\tomography$ is polynomial in $1/\delta, d$ and $\secp$.
\end{corollary}

%% file: detextraction.tex
\section{Deterministically Extracting Classical Strings from Quantum States}
\label{section:Extractor}
\noindent In this section, we show how to pseudodeterministically extract classical strings from $O(\log(\secp))$-qubit quantum states in \emph{polynomial} time. We first present the outline of our construction.
\begin{enumerate}
    \item Take as input $t(\secp)$ copies of a $d(\secp)$-dimensional (possibly mixed) quantum state $\rho$. Note that for our applications, we require $d(\secp) = \poly(\secp)$ and $t(\secp) = \poly(\secp)$.
    \item Perform $\tomography$ on the input $\rho^{\otimes t(\secp)}$ to get an approximation $M\in\C^{d\by d}$ of its classical description.
    \item Pick the first $k = o(d)$ diagonal entries of $M$, denoted by $p_1,\dots,p_k$. Divide them into $\ell$ groups where each of them is of size $r$ (namely, $k = \ell\cdot r$).
    \item In each group, consider the sum of all the elements. By $q_i$ we denote the sum of the $i$-th group.
    \item For each $q_i$, we round it to a bit, called $b_i$, according to which side it deviates from $r/d$.
    \item Output the concatenation of every bit $b_1||\dots||b_\ell$.
\end{enumerate}
In particular, we are interested in the case where the input is (polynomially many copies of) a \emph{Haar state}. Informally, a Haar state can be thought of as a uniformly random point on a high-dimensional sphere. We can partition the sphere into many regions and assign each region a unique bitstring. Given the input quantum state, the goal of the extractor is to find the corresponding bitstring. Hence, Haar states can be viewed as a natural source of randomness. Below, we present our main theorem.

\begin{theorem}
\label{thm:det:extraction}
There exists a quantum algorithm $\ext$ such that for all $d(\cdot)$, there exists a  (deterministic) function $f:\cD(\C^{d(\secp)})\to\bit^{\ell(\secp)}$ associated with $\ext$, where $\ell(\secp) = \floor{d(\secp)^{1/6}}$. On input $t(\secp) = \poly(d(\secp),\secp)$ copies of a $d(\secp)$-dimensional density matrix $\rho$, the algorithm $\ext$ outputs an $\ell(\secp)$-bit string $y$ and satisfies the following conditions.
\begin{itemize}
\item {\bf Efficiency:} The running time is polynomial in $d$ and $\secp$.
\item {\bf Correctness:}
For all $\secp\in\N$, there exists a set $\cG_\Delta\subseteq \cS(\C^{d(\secp)})$ such that 
\begin{enumerate}
    \item $\Pr\left[\ket{\psi}\in\cG_\Delta: \ket{\psi}\gets\Haar(\C^{d(\secp)})\right]
    \ge 1 - O(d(\secp)^{-1/6}).$
    \item For all $\ket{\psi}\in\cG_\Delta$,
        \begin{align*}
            \Pr\left[ y = f(\ketbra{\psi}{\psi}): y \gets \ext \left( \ketbra{\psi}{\psi}^{\otimes t(\secp)} \right) \right] 
            \ge 1 - \negl(\secp),
        \end{align*}
        where the probability is over the randomness of the extractor $\ext$.
\end{enumerate}

\item {\bf Statistical Closeness to Uniformity:} For sufficiently large $\secp\in\N$,
\begin{align*}
& d_{TV}( Y_\secp, U_{\ell(\secp)}) \le O(d(\secp)^{-1/6}),
\end{align*}
where $U_{\ell(\secp)}$ is the uniform distribution over all $\ell(\secp)$-bit strings and the random variable $Y_\secp$ is defined by the following process: 
\[
\ket{\psi} \gets \Haar(\C^{d(\secp)}), Y_\secp \gets \ext\left( \ketbra{\psi}{\psi}^{\otimes t(\secp)} \right).
\]
\end{itemize}
\end{theorem}

\begin{proof}
Here we present our construction of the extractor $\ext$.
\begin{construction}[The Extractor $\ext$] \hfill
\label{construction:extractor}
\begin{itemize}
    \item Input: $t(\secp) := 144\secp d(\secp)^8$ copies of a $d(\secp)$-dimensional quantum state $\rho\in \cD(\C^{d(\secp)})$.
    \item Perform $\tomography(\rho^{\otimes t(\secp)})$ with error tolerance $\delta(\secp) := d(\secp)^{-5/3}$ to get the classical description $M\in\C^{d(\secp)\by d(\secp)}$ that approximates $\rho$.
    \item Run $\round(M)$ to get $y\in\bit^{\ell(\secp)}$.
    \item Output $y$.
\end{itemize}
\end{construction}

\noindent The classical post-processing procedure $\round(M)$ is defined as follows: \\
\begin{mdframed}{\round($M$):}
\begin{itemize}
    \item Input: a matrix $M\in\C^{d(\secp)\by d(\secp)}$.
    \item Set parameters $k(\secp) := d(\secp)^{5/6}$, $r(\secp) := d(\secp)^{2/3}$ and $\ell(\secp) := d(\secp)^{1/6}$.
    \item Let $p_1,\dots,p_{d(\secp)}$ be the diagnal entries of $M$. For $i\in\set{1,\dots, \ell(\secp)}$, let $$q_i := \sum_{j=1}^r p_{(i-1)r + j}.$$
    \item For $i\in\set{1,\dots, \ell(\secp)}$, define
    \[
    b_i = 
    \begin{cases}
    0, & \text{ if } q_i < r/d \\
    1, & \text{ if } q_i > r/d.
    \end{cases}
    \]
    \item Output $b_1||\dots||b_{\ell(\secp)}$.
\end{itemize}
\end{mdframed}

\noindent By~\Cref{corollary:tomography} and the fact that $t(\secp) = \poly(d,\secp)$ and $\delta(\secp) = 1/\poly(d)$, it is easy to see that the running time of the extractor $\ext$ is polynomial in $d$ and $\secp$. Before proving the correctness and statistical closeness to uniformity, we present several statistical properties.

First, the distribution of the real and imaginary parts of any $k = o(d)$ coordinates of a Haar state $\ket{\psi}\sim\Haar(\C^d)$ is statistically close to a random vector with \iid Gaussian entries.
\begin{claim}
\label{claim:amplitudeToGaussian}
Let $\ket{\psi} = \sum_{i=1}^{d}\alpha_i\ket{i}$ be a uniformly random point on $\cS(\C^{d})$. Then
\[
d_{TV}\left((\Re(\alpha_1), \Im(\alpha_1), \dots, \Re(\alpha_k), \Im(\alpha_k)), Z \right)
= O\left(k/d\right),
\]
where $Z$ is a random variable in $\R^{2k}$ with \iid Gaussian entries $\cN(0,1/(2d))$.
\end{claim}
\begin{proof}
It immediately follows from \Cref{cor:Haar_Gaussian}.
\end{proof}

\noindent Next, since the $i$-th diagonal entry $p_i$ of $\ket{\psi}\bra{\psi}$ is the squared absolute value of the $i$-th coordinate $\alpha_i$ of $\ket{\psi}$, the joint distribution of $(p_1,\dots,p_k)$ is statistically close to a random vector in $\R^k$ with \iid $\chi_2^2$ entries.
\begin{claim}
\label{claim:dist:diagonal}
$d_{TV}\left( (p_1, \dots, p_k), Q/(2d) \right) = O(k/d)$ where $Q$ is a random variable in $\R^k$ with \iid $\chi^2_2$ entries.
\end{claim}
\begin{proof}
For $i\in[k]$, each diagonal entry $p_i = |\alpha_i|^2 = \Re(\alpha_i)^2 + \Im(\alpha_i)^2$. By~\Cref{claim:amplitudeToGaussian}, the total variation distance induced by replacing the real and imaginary parts of the amplitudes with \iid Gaussians $\cN(0,1/(2d))$ is $O(k/d)$. Then by setting $f(x_1,\dots, x_{2k}) := \left(x_1^2+x_2^2, \dots, x_{2k-1}^2 + x_{2k}^2\right)$ in~\Cref{fact:dataProcessingIneq} and the definition of $\chi^2_2$, we complete the proof.
\end{proof}

\noindent Now, we consider the distribution of $q_i$'s. Note that the sum of $r$ \iid $\chi_2^2$ random variables is identically distributed to $\chi_{2r}^2$ by the property of the $\chi^2$-distribution in~\Cref{fact:chi-squared}. Namely, the joint distribution of $(q_1, \dots, q_\ell)$ is statistically close to a random vector in $\R^\ell$ with \iid $\chi_{2r}^2$ entries.

\begin{claim}
\label{claim:dist:keybits}
$d_{TV}\left( (q_1, \dots, q_\ell), R/(2d) \right) = O(k/d)$ where $R$ is a random variable in $\R^\ell$ with \iid $\chi^2_{2r}$ entries.
\end{claim}
\begin{proof}
Recall that $q_i := \sum_{j=1}^r p_{(i-1)r + j}$. From~\Cref{claim:dist:diagonal}, we have $d_{TV}\left( (p_1, \dots, p_k), Q/(2d) \right) = O(k/d)$, where $Q = (Q_1, \dots, Q_k)$ is a random variable in $\R^k$ with \iid $\chi^2_2$  entries. Hence by the data processing inequality~(\Cref{fact:dataProcessingIneq}) and setting 
\[
f(x_1, \dots, x_k) := \left( \sum_{j = 1}^r x_{j}, \sum_{j = 1}^r x_{r+j}, \dots, \sum_{j=1}^r x_{(\ell - 1)r + j} \right),
\]
we have $d_{TV}\left( (q_1, \dots, q_\ell), R/(2d) \right) = O(k/d)$, where we use the fact that the sum of $r$ \iid $\chi^2_2$ random variables is identically distributed to $\chi^2_{2r}$.
\end{proof}

\noindent Moreover, a $\chi^2_{2r}$ random variable is the sum of $r$ \iid absolutely continuous random variables. Hence, relying on the aforementioned central limit theorem, it is statistically close to a Gaussian distribution.
\begin{lemma}
\label{lem:stringsGaussian}
$d_{TV}\left((q_1, \dots, q_\ell), Z/(2d) \right) = O(k/d) + O(\ell/\sqrt{r})$ where $Z$ is a random variable in $\R^\ell$ with \iid $\cN(2r, 4r)$ entries, i.e., $Z/(2d)$ has \iid $\cN(r/d, r/d^2)$ entries.
\end{lemma}
\begin{proof}
By~\Cref{cor:CLTofChi} and hybrids over every coordinate for $i\in[\ell]$, we have $d_{TV}(R/(2d), Z/(2d)) = O(\ell/\sqrt{r})$, where $R$ is defined in~\Cref{claim:dist:keybits}. Together with~\Cref{claim:dist:keybits} finishes the proof.
\end{proof}

\noindent Now, we are ready to prove the correctness and the statistical closeness to uniform properties. \\

\noindent \textbf{Correctness.} First, define the function $f:\cD(\C^{d(\secp)})\to\bit^{\ell(\secp)}$ associated with the extractor as $$f(\sigma) := \round(\sigma).$$

\noindent Due to the continuous nature of quantum states, it is impossible to discretize them perfectly. For any $\sigma\in\cD(\C^d)$, consider the corresponding $q_1,\dots,q_\ell$ defined in~\Cref{construction:extractor}. If \emph{all} $q_1,\dots,q_\ell$ are sufficiently away from $r/d$, then the extractor is able to output the correct string with high probability by the correctness of $\tomography$. 
Here, we define the set $\cG_\Delta$ of ``good states'' whose $q_1,\dots,q_\ell$ are \emph{all} $\Delta$-away from $r/d$ (the parameter $\Delta(\secp)$ will be chosen later). The following claim characterizes the probability of a Haar random state being in $\cG_\Delta$.

\begin{claim}
\label{claim:HaarGood}
Let the set $\cG_\Delta\subseteq \cS(\C^{d(\secp)})$ be
\begin{align*}
    \cG_\Delta := \set{\ket{\psi}\in \cS(\C^{d(\secp)}): \forall i\in[\ell],\ \left|q_i-\frac{r}{d}\right| > \Delta },
\end{align*}
where each $q_i$ is defined on the matrix $\ketbra{\psi}{\psi}$.
It holds that
\begin{align*}
    \Pr\left[\ket{\psi}\in\cG_\Delta: \ket{\psi}\gets\Haar(\C^{d(\secp)})\right]
    \ge 1 - O\left(\frac{k}{d}\right) - O\left(\frac{\ell}{\sqrt{r}}\right) - O\left(\frac{\Delta\ell d}{\sqrt{r}}\right).
\end{align*}
\end{claim}
\begin{proof}
By \Cref{lem:stringsGaussian}, the total variation distance between $(q_1,\dots,q_\ell)$ and the random variable $Z = (Z_1,\dots,Z_\ell)$ with \iid Gaussian entries $Z_i \sim \cN(r/d, r/d^2)$ is $O(k/d) + O(\ell/\sqrt{r})$. Hence, 
\begin{align*}
    \Pr\left[ \ket{\psi}\in\cG_\Delta: \ket{\psi}\gets\Haar(\C^{d(\secp)}) \right] 
    & =  \Pr\left[ \forall i\in[\ell],\ \left|q_i-\frac{r}{d}\right| > \Delta \right] \\
    & \ge \Pr\left[ \forall i\in[\ell],\ \left| Z_i - \frac{r}{d} \right| > \Delta \right] 
    - O\left(\frac{k}{d}\right) - O\left(\frac{\ell}{\sqrt{r}}\right).
\end{align*}
Moreover, by~\Cref{fact:NormalDist}, for every coordinate $i\in[\ell]$, it holds that 
\[
\Pr\left[ \left| Z_i - \frac{r}{d} \right| \le \Delta \right]
\le O\left(\frac{\Delta}{\sqrt{r}/d}\right).
\]
\noindent By a union bound over $i \in [\ell]$, with all but $O\left(\frac{\Delta\ell d}{\sqrt{r}}\right)$ probability every $Z_i$ is $\Delta$-away from $r/d$. Collecting the probabilities completes the proof of~\Cref{claim:HaarGood}.
\end{proof}

\noindent Hence, by setting $\Delta(\secp) = 1/d(\secp)$, the choice of parameters $r(\secp) = d(\secp)^{2/3}$, $\ell(\secp) = d(\secp)^{1/6}$, $k(\secp) = d(\secp)^{5/6}$ and~\Cref{claim:HaarGood}, we have
\begin{align*}
\Pr\left[ \ket{\psi} \in \cG_\Delta: \ket{\psi}\gets\Haar(\C^{d(\secp)}) \right] 
\ge 1 - O(d(\secp)^{-1/6}).
\end{align*}

Next, given a state which is in $\cG_\Delta$, the output bitstring extracted from it will be $f(\ketbra{\psi}{\psi})$ with overwhelming probability by the correctness of $\tomography$ in~\Cref{corollary:tomography}.
\begin{claim}
\label{claim:GoodCorrect}
If $\ket{\psi}\in\cG_\Delta$, then running $\ext$ in~\Cref{construction:extractor} with error tolerance $\delta(\secp) = d^{-5/3} = \Delta(\secp)/r(\secp)$ for $\tomography$ satisfies
\[
\Pr\left[ y = f(\ketbra{\psi}{\psi}): y \gets \ext \left( \ketbra{\psi}{\psi}^{\otimes t(\secp)} \right) \right] 
\ge 1 - \negl(\secp),
\]
where the probability is over the randomness of the extractor $\ext$.
\end{claim}
\begin{proof}
Let $M$ be the classical description obtained by running $\tomography(\ketbra{\psi}{\psi}^{\otimes t(\secp)})$ with error tolerance $\delta$ and $t \ge 144\secp d^8 \ge 144\secp d^4/\delta^2$. Let $\hat{p}_i$'s and $\hat{q}_j$'s be the corresponding diagonal entries and sums of $M$. By~\Cref{corollary:tomography} , $\TD(\ketbra{\psi}{\psi}, M)\le \delta$ holds with overwhelming probability. For the rest of the proof, we assume that this event holds. Then by~\Cref{fact:TDdiagonal}, we have $|p_i - \hat{p}_i|\le \delta$ for every $i\in[k]$. Since $\ket{\psi}\in\cG_\Delta$, we have $|q_i - r/d| > \Delta$ for every $i\in[\ell]$. We now show that if $q_i > r/d + \Delta$, then $\hat{q}_i > r/d$. For every $i\in[\ell]$, by the triangle inequality and the fact that $\delta = \Delta/r$, we have
\[
\hat{q}_i
= q_i - (q_i - \hat{q}_i)
> \left( \frac{r}{d} + \Delta \right) - \sum_{j=1}^r \left|p_{(i-1)r+j} - \hat{p}_{(i-1)r + j}\right|
\ge \frac{r}{d} + \Delta - r\cdot \delta = \frac{r}{d}.
\]
Similarly, we have $q_i < r/d - \Delta$ implies that $\hat{q}_i < r/d$. Hence, this ensures the consistency between $\round(M)$ and $\round(\ketbra{\psi}{\psi})$ and completes the proof of~\Cref{claim:GoodCorrect}.
\end{proof}

\paragraph{\bf{Statistical Closeness to Uniformity.}}
We finish the proof with a hybrid argument:
\begin{itemize}
    \item $\hybrid_1:$ In the first hybrid, the output is generated according to~\Cref{construction:extractor}.
    \begin{enumerate}
        \item Sample $\ket{\psi}\gets\Haar(\C^{d(\secp)})$.
        \item Perform $\tomography(\rho^{\otimes t(\secp)})$ with $t(\secp) := 144\secp d(\secp)^8$ and error tolerance $\delta(\secp) := d(\secp)^{-5/3}$ to get the classical description $M\in\C^{d(\secp)\by d(\secp)}$ that approximates $\rho$.
        \item Output $y = \round(M)$.
    \end{enumerate}
    
    \item $\hybrid_2:$ In the second hybrid, the input of $\round$ is changed to the exact description of the quantum state.
    \begin{enumerate}
        \item Sample $\ket{\psi}\gets\Haar(\C^{d(\secp)})$.
        \item Output $y = \round(\ketbra{\psi}{\psi})$.
    \end{enumerate}
     
    \item $\hybrid_3:$ In the third hybrid, the output is generated by rounding \iid Gaussians.
     \begin{enumerate}
        \item Sample $z_1,\dots,z_\ell \gets \cN(r/d, r/d^2)$.
        \item For $i\in[\ell]$, 
        \[
            b_i = 
            \begin{cases}
                0, & \text{ if } z_i < r/d \\
                1, & \text{ if } z_i > r/d.
            \end{cases}
        \]
    \item Output $b_1||\dots||b_{\ell(\secp)}$.
    \end{enumerate}
\end{itemize}

\noindent We can bound the total variation distance between $\hybrid_1$ and $\hybrid_2$ by $O(\delta) = O(d^{-5/3})$ using~\Cref{corollary:tomography} and our chosen error tolerance. Additionally, the total variation distance between $\hybrid_2$ and $\hybrid_3$ is at most $O(d^{-1/6})$ from~\Cref{lem:stringsGaussian}. Finally, since Gaussians are symmetric about the mean, the output string in $\hybrid_3$ is uniformly and randomly distributed. This finishes the proof of~\Cref{thm:det:extraction}.
\end{proof}

%% file: qpseud.tex
\section{Quantum PRGs and PRFs}
\label{section:Quantum Pseudorandomess}

In this section, we present our main application of the extractor in~\Cref{section:Extractor}. We introduce the notion of \emph{pseudodeterministic quantum pseudorandom generators} (QPRGs). As the name suggests, QPRGs is a pseudorandom generator with quantum generation satisfying only \emph{pseudodeterminism} property. To be more precise, by pseudodeterminism we mean that there exist some constant $c>0$ and at least $1 - O(\secp^{-c})$ fraction of ``good seeds'' for which the output is almost certain. That is, for each good seed, the probability (over the randomness of the QPRG) of the most likely output is at least $1 - O(\secp^{-c})$. 

\subsection{Construction of QPRGs}

\begin{definition}[Weak/Strong Pseudodeterministic Quantum Pseudorandom Generator] 
A weak pseudodeterministic quantum pseudorandom generator $G_{\secp}$, abbreviated as wQPRG, is a uniform QPT algorithm that on input a seed $k \in \bit^\secp$, outputs a bitstring of length $\ell(\secp)$ with the following guarantees:
\begin{itemize}
\item {\bf Pseudodeterminism:} There exists a constant $c > 0$ and a function $\mu(\secp) = O(\secp^{-c})$ such that for every $\secp\in\N$, there exists a set of ``good seeds'' $\cK_\secp \subseteq \bit^\secp$ satisfying the following:
    \begin{enumerate}
        \item $\Pr[k\in\cK_\secp: k \gets \bit^\secp] \ge 1 - \mu(\secp)$.
        \item For any $k\in\cK_\secp$, it holds that 
        \[
        \max\limits_{y\in\bit^{\ell(\secp)}} \Pr[y = G_\secp(k)] \ge 1 - \mu(\secp),
        \]
        where the probability is over the randomness of $G_\secp$.
    \end{enumerate}

\item {\bf Stretch:} The output length of $G_{\secp}$, namely $\ell(\secp)$, is strictly greater than $\secp$.

\item {\bf Weak Security:} For every (non-uniform) QPT distinguisher $A$, there exists a polynomial $\nu(\cdot)$ such that the following holds for sufficiently large $\secp\in\N$,
\begin{equation} \label{equaiton:distinguishingAdv}
\left| \Pr\left[ A_\secp(y) = 1: \substack{ k \gets \bit^\secp, \\
y \gets G_\secp(k)} \right] 
- \Pr\left[ A_\secp(y) = 1: y \gets \bit^{\ell(\secp)} \right] \right| 
\le 1 - \frac{1}{\nu(\secp)},
\end{equation}
where the probability of the first experiment is over the choice of $k$ and the randomness of $G_\secp$ and $A_\secp$.

\noindent If $G$ further satisfies the strong security property defined below, we call $G$ a \emph{strong pseudodeterministic quantum pseudorandom generator}, abbreviated as sQPRG.
\item {\bf Strong Security:} For every (non-uniform) QPT distinguisher $A$, there exists a negligible function $\veps(\cdot)$ such that the following holds for sufficiently large $\secp\in\N$,
\[
\left| \Pr\left[ A_\secp(y) = 1: \substack{ k \gets \bit^\secp, \\
y \gets G_\secp(k)} \right] 
- \Pr\left[ A_\secp(y) = 1: y \gets \bit^{\ell(\secp)} \right] \right| 
\le \veps(\secp),
\]
where the probability of the first experiment is over the choice of $k$ and the randomness of $G_\secp$ and $A_\secp$.

\end{itemize}
We call the left-hand side of~\Cref{equaiton:distinguishingAdv} the \emph{distinguishing advantage} of $A$. We say that $G$ is \emph{$(1-\delta(\secp))$-pseudorandom} or \emph{has pseudorandomness $1-\delta(\secp)$} if the maximum distinguishing advantage over all non-uniform QPT adversaries is at most $\delta(\secp)$. We say that $G$ has pseudodeterminism $1 - \mu(\secp)$ if it satisfies the pseudodeterminism property for the function $\mu(\cdot)$. 
\end{definition}

\noindent We begin with an $n(\secp)$-PRS (recall that $n(\secp)$ is its output length), where $n(\secp) = O(\log\secp)$ and the dimension of its output is $d(\secp) = 2^{n(\secp)} = \poly(\secp)$.

\begin{theorem}[$O(\log\secp)$-PRS implies wQPRG]
\label{lemma:QPRG:determinism}
Assuming the existence of $(c\log\secp)$-PRS for some constant $c > 6$, then there exists a $\left(1 - O(\secp^{-c/6})\right)$-pseudorandom wQPRG with pseudodeterminism $1 - O(\secp^{-c/12})$ and output length $\ell(\secp) = \secp^{c/6} > \secp$.
\end{theorem}
\begin{proof}
Consider the following construction.
\begin{construction}[Weak Quantum Pseudorandom Generators] \hfill
\label{construction:wQPRG}
\begin{enumerate}
\item Input: a security parameter $1^\secp$ and a seed $k\in\bit^\secp$.
\item Run $(c\log\secp)$-$\prs(k)$ $t$ times to get $\rho_k^{\otimes t(\secp)}$, where $t(\secp) = 144\secp d(\secp)^8 = O(\secp^{8c+1})$ as defined in~\Cref{construction:extractor}.
\item Run $\ext( \rho_k^{\otimes t(\secp)} )$ defined in~\Cref{construction:extractor} to get $y\in\bit^{\ell(\secp)}$.
\item Output $y$.
\end{enumerate}
\end{construction}

\paragraph{Efficiency.} Since $t(\secp) = \poly(\secp)$ and $d(\secp) = O(\secp^c)$, the running time is polynomial in $\secp$ from~\Cref{thm:det:extraction}.

\paragraph{Pseudodeterminism.} We complete the proof by a hybrid argument. Consider the following hybrids.
\begin{itemize}
    \item $\hybrid_1:$ In the first hybrid, $y$ is generated according to~\Cref{construction:wQPRG}.
    \begin{enumerate}
        \item Sample $k\gets\bit^\secp$. 
        \item Run $\prs(k)$ $t$ times to get $\rho_k^{\otimes t(\secp)}$.
        \item Run $y\gets\ext(\rho_k^{\otimes t(\secp)})$.
        \item Output $y$.
    \end{enumerate} 
    
    \item $\hybrid_2:$ In the second hybrid, the input is changed to a Haar state.
    \begin{enumerate}
        \item Sample $\ket{\psi}\gets\Haar(\C^{d(\secp)})$.
        \item Run $y\gets\ext(\ketbra{\psi}{\psi}^{\otimes t(\secp)})$.
        \item Output $y$. 
    \end{enumerate}
\end{itemize}

\noindent For the sake of contradiction, suppose there exists at least $\mu(\secp)\neq O(\secp^{-c/12})$ fraction of ``bad seeds'' (the complement of the set $\cK_\secp$ of good seeds) for which the probability of the most likely output is at most $1 - \mu(\secp)$.
Then we construct an efficient distinguisher for $\prs$ as follows:
\begin{enumerate}
    \item Take as input $2t(\secp) = \poly(\secp)$ copies of $\rho$ which is either sampled from $\prs$ with a random key or $\Haar(\C^{d(\secp)})$.
    \item Run $\ext(\rho^{\otimes t(\secp)})$ twice independently and get the output $y_1, y_2$ respectively.
\end{enumerate}

\noindent First, if $\rho$ is sampled from $\Haar(\C^{d(\secp)})$, then by the correctness of $\ext$ in~\Cref{thm:det:extraction}, we have 
\[
p_1 := \Pr\left[ y_1 = y_2: \rho\gets\Haar(\C^{d(\secp)}) \right]
\ge (1 - O(d(\secp)^{1/6})) \cdot (1-\negl(\secp))^2 
\ge 1 - h(\secp),
\]
where $h(\secp) = O(\secp^{-c/6})$.

\noindent On the other hand, consider the case in which $\rho$ is sampled from $\prs$. Without loss of generality, we can assume that $\mu(\secp)<1/2$ for sufficiently large $\secp$. Otherwise, the distinguishing advantage would already be non-negligible. Then,
\begin{align*}
p_2 & := \Pr\left[ y_1 = y_2: \rho\gets\prs(k) \right]
= \Pr[k \in \cK_\secp]\Pr[y_1 = y_2 \mid k \in \cK_\secp] + \Pr[k \notin \cK_\secp]\Pr[y_1 = y_2 \mid k \notin \cK_\secp] \\
& \le (1 - \mu(\secp)) \cdot 1 + \mu(\secp) \cdot ( 1 - \mu(\secp) )
= 1 - \mu(\secp)^2.
\end{align*}
Now, for any sufficiently large $\secp\in\N$ such that $1/2 > \mu(\secp)$, we do a case analysis. Suppose $h(\secp) \ge \mu(\secp)^2$, then we have $\sqrt{h(\secp)} > \mu(\secp)$. Otherwise, if $h(\secp) < \mu(\secp)^2$, then the distinguishing advantage $|p_1 - p_2|$ satisfies
\[
\negl(\secp) 
= |p_1 - p_2| = p_1 - p_2 
\ge \mu(\secp)^2 - h(\secp)
\]
due to the security of $\prs$. Hence, it holds that $\sqrt{h(\secp)+ \negl(\secp)} > \mu(\secp)$. However, combining two cases would imply that $\mu(\secp) = O(\secp^{-c/12})$ and lead to a contradiction.

\paragraph{Stretch.} From~\Cref{thm:det:extraction}, the output length of~\Cref{construction:wQPRG} is given by $\ell(\secp) = d(\secp)^{1/6} = \secp^{c/6} > \secp$ since $c > 6$.

\paragraph{Weak Security.} We complete the proof by a hybrid argument. Consider the following hybrids:
\begin{itemize}
    \item $\hybrid_1:$ In the first hybrid, the adversary receives a string $y$ which is generated according to~\Cref{construction:wQPRG}.
    \begin{enumerate}
        \item Sample $k\gets\bit^\secp$.
        \item Run $\prs(k)$ $t$ times to get $\rho_k^{\otimes t(\secp)}$.
        \item Run $y\gets\ext(\rho_k^{\otimes t(\secp)})$.
        \item Output $y \in\bit^{\ell(\secp)}$. 
    \end{enumerate}
    
    \item $\hybrid_2:$ In the second hybrid, the input is changed to a Haar state. 
    \begin{enumerate}
        \item Sample $\ket{\psi}\gets\Haar(\C^{d(\secp)})$.
        \item Run $y\gets\ext(\ketbra{\psi}{\psi}^{\otimes t(\secp)})$.
        \item Output $y \in\bit^{\ell(\secp)}$.
    \end{enumerate}
    
    \item $\hybrid_3:$ Sample $y\gets\bit^{\ell(\secp)}$. Output $y \in\bit^{\ell(\secp)}$.
    In the third hybrid, the adversary receives a string sampled from the uniform distribution. 
\end{itemize}

\noindent The computational indistinguishability of hybrids $\hybrid_1$ and $\hybrid_2$ follows from the security of $\prs$. Otherwise, running $\ext$ on the samples once would be an efficient distinguisher in the PRS security experiment. The statistical indistinguishability of hybrids $\hybrid_2$ and $\hybrid_3$ follows from the statistical closeness to uniform property of~\Cref{thm:det:extraction}. In particular, the statistical distance is $O(d(\secp)^{-1/6}) = O(\secp^{-c/6})$.
\end{proof}

While we do not have a non-trivial way to amplify the pseudodeterminism property, the security amplification can be achieved by techniques in~\cite{CHS05TCC, DIJK09TCC, MauTes09CRYPTO, MauTes10TCC}. In particular, we will use the security amplification for (classical) weak PRGs in~\cite{DIJK09TCC}. The construction is to run the weak PRG $G$ on input $s(\secp) = \omega(\log\secp)$ independently and randomly chosen seeds $k_1, \dots, k_s$ and then output the bit-wise XOR of the $s$ strings $G(k_1), \dots, G(k_s)$.

\begin{theorem}[{\cite[Theorem 6]{DIJK09TCC}}]
\label{thm:amplification:PRG}
Let $s(\secp) = \omega(\log\secp)$. Let $G:\bit^\secp \to \bit^{\ell(\secp)}$ be a weak PRG with $(1-\delta)$-pseudorandomness such that $\delta < 1/2$ and $\ell(\secp) > s(\secp)\cdot\secp$. Define the function $G^{\xor s}: \bit^{s(\secp)\secp} \to \bit^{\ell(\secp)}$ as $G^{\xor s}(k_1,\dots,k_s) := \bigxor_{i=1}^s G(k_i)$. Then $G^{\xor s}$ is a strong PRG. 
\end{theorem}

\noindent We observed that~\Cref{thm:amplification:PRG} could be extended to QPRGs.

\begin{theorem}[Security amplification for QPRGs]
\label{thm:amplification:QPRG}
Let $G:\bit^\secp \to \bit^{\ell(\secp)}$ be a wQPRG that has pseudodeterminism $1 - O(\secp^{-c})$ and pseudorandomness $1 - \delta$ such that $c>1$, $\delta(\secp) \le 0.49+o(1)$ and $\ell(\secp) > s(\secp)\cdot \secp$, where $s(\secp) = \Theta(\secp)$. Define the QPT algorithm $G^{\xor s}: \bit^{s(\secp)\secp} \to \bit^{\ell(\secp)}$ as $G^{\xor s}(k_1,\dots,k_s) := \bigxor_{i=1}^s G(k_i)$. Then $G^{\xor s}$ is a sQPRG with pseudodeterminism $1 - O(\secp^{-(c-1)})$ and output length $\ell(\secp)$.
\end{theorem}
\begin{proof}[Proof sketch]
We first prove the pseudodeterminism property. Fix the security parameter $\secp$. The set of good seeds of $G^{\xor s}$ is defined to be the $s(\secp)$-fold Cartesian product $\cK_\secp \times \dots \times \cK_\secp \subseteq \bit^{s(\secp)\secp}$, where $\cK_\secp$ is the set of good seeds of $G$. By the pseudodeterminism of $G$ and a union bound over $i\in[s]$, we have 
\[
\Pr\left[ \forall i\in[s],\ k_i\in\cK_\secp: k = k_\secp||\dots||k_\secp \gets \bit^{s(\secp)\secp} \right] \ge 1 - O(s/\secp^c) = 1 - O(\secp^{-(c-1)}).
\]
Next, for every $(k_1, \dots, k_s)  \in \cK_\secp \times \dots \times \cK_\secp$ and every $i\in[s]$, there exists some $y_i\in\bit^{\ell(\secp)}$ such that $\Pr[y_i = G(k_i)] \ge 1 - 1/\secp^c$ for every $i\in[s]$. Hence, by a union bound over $i\in[s]$, it holds that 
\[
\Pr\left[ \bigxor_{i=1}^s y_i = G^{\xor s}(k_1,\dots,k_s) \right]
\ge \Pr\left[ \bigwedge_{i=1}^s y_i = G(k_i) \right]
\ge 1 - O(s/\secp^c) = 1 - O(\secp^{-(c-1)}).
\]
That is, $G^{\xor s}$ has pseudodeterminism $1 - O(\secp^{-(c-1)})$. The full proof of (strong) security is deferred to~\Cref{app:amplification:PRG}.
\end{proof}

\noindent From~\Cref{lemma:QPRG:determinism}, \Cref{thm:amplification:QPRG} and picking $s(\secp) = \secp$, we have the following corollary.
\begin{corollary}
Assuming the existence of $(c\log\secp)$-PRS for some constant $c > 12$, then there exists a sQPRG $G^{\xor \secp}: \bit^{\secp^2} \to \bit^{\ell(\secp)}$ with pseudodeterminism $1 - O(\secp^{-(c/12-1)})$ and output length $\ell(\secp) = \secp^{c/6} > \secp^2$.
\end{corollary}

\subsection{Construction of Selectively Secure QPRFs}

In the same spirit, it is natural to consider the concept of pseudodeterministic quantum pseudorandom functions (QPRFs). However, when the pseudodeterminism is only $1 - O(\lambda^{-c})$, there is a caveat. An attacker that can make \emph{adaptive} queries can easily distinguish a QPRF with this level of pseudodeterminism from a random function as follows: the distinguisher simply queries the oracle on \emph{the same point} polynomially many times and checks if the answers are all the same. A random function will always produce identical outputs, while a QPRF with pseudodeterminism $1-O(\lambda^{-c})$ will generate different outputs with constant probability. Intuitively, non-determinism allows the QPRF output to appear more random, thus it should strengthen its security. 

Below, we show that we can use a selectively secure $(m(\secp),n(\secp))$-PRFS, where $m(\secp) = \omega(\log\secp)$ and $n(\secp) = O(\log\secp)$, to construct a selectively secure QPRF with input length $m(\secp)$ and output length $\poly(\secp)$. 

\begin{definition}[Selectively Secure Quantum Pseudorandom Functions]
A \emph{selectively secure quantum pseudorandom function} $F:\bit^\secp \times \bit^{m(\secp)} \to \bit^{\ell(\secp)}$ is a QPT algorithm with the following guarantees: 
\begin{itemize}

\item {\bf Pseudodeterminism:} There exists a constant $c > 0$ and a function $\mu(\secp) = O(\secp^{-c})$ such that for every $\secp\in\N$ and every $x\in\bit^{m(\secp)}$, there exists a set of ``good keys'' $\cK_{\secp,x} \subseteq \bit^\secp$ satisfying the following:
\begin{enumerate}
    \item $\Pr[k \in \cK_{\secp,x}: k \gets \bit^\secp] \ge 1 - \mu(\secp)$.
    \item For any $k \in \cK_{\secp,x}$, it holds that 
    \[
    \max\limits_{y\in\bit^{\ell(\secp)}} \Pr[y = F(k,x)] \ge 1 - \mu(\secp),
    \]
    where the probability is over the randomness of $F$.
\end{enumerate}

\item {\bf Selective Security:} 
For any polynomial $q(\cdot)$, any (non-uniform) QPT distinguisher $A$ and any family of pairwise distinct indices $\left( \set{x_1, \dots, x_{q(\secp)}} \subseteq \bit^{m(\secp)}\} \right)_\secp$, there exists a negligible function $\veps(\cdot)$ such that for all $\secp\in\N$,
\begin{multline*}
\Bigg| \Pr\left[ A_\secp(x_1, \dots, x_{q{(\secp)}}, y_1, \dots, y_{q(\secp)} ) = 1:
\substack{ k\gets\bit^{\secp}, \\ 
y_1 \gets F(k,x_1), \dots, y_{q{(\secp)}} \gets F(k,x_{q{(\secp)}}) } \right] \\
- \Pr\left[ A_\secp(x_1, \dots, x_{q{(\secp)}}, y_1, \dots, y_{q(\secp)} ) = 1:
y_1,\dots,y_{q{(\secp)}} \gets\bit^{\ell(\secp)} \right] \Bigg|
\le \veps(\secp).
\end{multline*}

\end{itemize}
\end{definition}

\noindent We will construct a selectively secure quantum pseudorandom function based on a selectively secure $(m(\secp),n(\secp))$-PRFS, where $m(\secp) = \omega(\log\secp)$ and $n(\secp) = O(\log\secp)$.

\begin{theorem}[$(\omega(\log\secp),O(\log\secp))$-PRFS implies selectively secure QPRF]
\label{theorem:QPRF}
Assuming the existence of selectively secure $(m(\secp), c\log\secp)$-PRFS for some constant $c > 12$ and $m(\secp) = \omega(\log\secp)$, then there exists a selectively secure QPRF $F:\bit^{\secp^2}\times\bit^{m(\secp)} \to \bit^{\ell(\secp)}$ with pseudodeterminism $1 - O(\secp^{-(c/12-1)})$, input length $m(\secp)$ and output length $\ell(\secp) = \secp^{c/6}$.
\end{theorem}
\begin{proof}

Consider the following construction:
\begin{construction}[Selectively Secure Quantum Pseudorandom Functions] \hfill
\label{construction:QPRF}
    \begin{enumerate}
        \item Input: a key $k\in\bit^{\secp^2}$ and input $x\in\bit^{m(\secp)}$.
        \item Parse $k$ as $k_1||\dots||k_\secp$ such that $k_i \in \bit^\secp$ for every $i\in[\secp]$.
        \item For $i\in[\secp]$, run $\prfs(k_i, x)$ to get $\rho_{k_i, x}^{\otimes t(\secp)}$, where $t(\secp) = 144\secp d(\secp)^8$.
        \item For $i\in[\secp]$, run $\ext(\rho_{k_i, x}^{\otimes t(\secp)})$ to get $y_i\in\bit^{\ell(\secp)}$.
        \item Let $y = \bigxor_{i=1}^\secp y_i$, output $y\in\bit^{\ell(\secp)}$.
    \end{enumerate}
\end{construction}

\paragraph{Pseudodeterminism.}
We complete the proof by a hybrid argument. For any fixed $x\in\bit^{m(\secp)}$, consider the following hybrids.

\begin{itemize}
    \item $\hybrid_1:$ In the first hybrid, $y$ is generated according to~\Cref{construction:QPRF}.
    \begin{enumerate}
        \item Sample $k\gets\bit^{\secp^2}$.
        \item Parse $k$ as $k_1||\dots||k_\secp$ such that $k_i \in \bit^\secp$ for every $i\in[\secp]$.
        \item For $i\in[\secp]$, run $\prfs(k_i, x)$ $t$ times to get $\rho_{k_i, x}^{\otimes t(\secp)}$.
        \item For $i\in[\secp]$, run $\ext(\rho_{k_i, x}^{\otimes t(\secp)})$ to get $y_i\in\bit^{\ell(\secp)}$.
        \item Let $y = \bigxor_{i=1}^\secp y_i$, output $y\in\bit^{\ell(\secp)}$.
    \end{enumerate}

    \item $\hybrid_2:$ In the second hybrid, the input of $\ext$ is changed to Haar states.
    \begin{enumerate}
        \item For $i\in[\secp]$, sample $\ket{\psi_i}\gets\Haar(\C^{d(\secp)})$.
        \item For $i\in[\secp]$, run $\ext( \ketbra{\psi_i}{\psi_i}^{\otimes t(\secp)} )$ to get $y_i\in\bit^{\ell(\secp)}$.
        \item Let $y = \bigxor_{i=1}^\secp y_i$, output $y\in\bit^{\ell(\secp)}$.
    \end{enumerate}
\end{itemize}

\noindent Similar to proving pseudodeterminism in~\Cref{lemma:QPRG:determinism}, there exists at least $1 - O(\secp^{-c/12})$ fraction of good keys $\cK'_{\secp,x}\subseteq \bit^{\secp}$ such that for any $k \in \cK'_{\secp,x}$, 
\[
\max\limits_{y\in\bit^{\ell(\secp)}} \Pr[y = \ext(\prfs(k, x)^{\otimes t})] 
\ge 1 - O(\secp^{-c/12}),
\]
where the probability is over the randomness of $\ext$. Otherwise, running $\ext$ independently twice on input $\prfs(k, x)^{\otimes t}$ and comparing the output would be an efficient distinguisher that contradicts the security of $\prfs$. The set of good keys for $F$ is defined to be the $\secp$-fold Cartesian product $\cK'_{\secp,x} \times \dots \times \cK'_{\secp,x} \subseteq \bit^{\secp^2}$. Then following the same lines for proving pseudodeterminism in~\Cref{thm:amplification:QPRG}, we can conclude that~\Cref{construction:QPRF} has pseudodeterminism $1 - O(\secp^{-(c/12-1)})$.

\paragraph{Selective Security.}
Before we prove the security, we introduce a simple lemma regarding the indistinguishability of polynomially many samples of Haar states and the output of a PRS generator with \iid uniform seeds. 

\begin{lemma} \label{lem:poly_copy_PRS}
Let $\prs$ be an $n(\cdot)$-PRS. Then for any polynomials $t(\cdot),p(\cdot)$ and any QPT distinguisher $A$, there exists a negligible function $\veps(\cdot)$ such that
\begin{multline*}
\Pr\left[ A\left( \rho^{\otimes t(\secp)}_1, \dots, \rho^{\otimes t(\secp)}_{q(\secp)} \right) = 1 : \substack{k_1,\dots,k_{q(\secp)} \gets \bit^\secp, \\
\rho_1 \gets \prs(k_1),\dots,\rho_{q(\secp)} \gets \prs(k_{q(\secp)})} 
\right] \\
- \Pr\left[ A\left( \ket{\vartheta_1}^{\otimes t(\secp)}, \dots, \ket{\vartheta_{q(\secp)}}^{\otimes t(\secp)} \right) = 1:
\ket{\vartheta_1}, \dots, \ket{\vartheta_{q(\secp)}} \gets \Haar_{n(\secp)}
\right] \leq \veps(\secp).
\end{multline*}
\begin{proof}
Consider the following hybrids $\hybrid_i$ for $i\in\set{0,1,\dots,q}$:
\begin{enumerate}
    \item For $1 \le j \le i$, sample $k_j \gets \bit^\secp$ and run $\prs(k_j)$ $t$ times to get $\rho^{\otimes t}_j$.
    \item For $i+1 \le j \le q$, sample $\ket{\vartheta_j} \gets \Haar_{n(\secp)}$.
    \item Output $\left( \rho^{\otimes t(\secp)}_1, \dots, \rho^{\otimes t(\secp)}_i, \ket{\vartheta_{i+1}}^{\otimes t(\secp)}, \dots, \ket{\vartheta_{q(\secp)}}^{\otimes t(\secp)} \right)$.
\end{enumerate}
It is sufficient to prove the computational indistinguishability between $\hybrid_i$ and $\hybrid_{i+1}$. Note that the only difference is the $(i+1)$-th coordinate of the sample. Suppose there exist polynomials $t(\cdot),q(\cdot)$ and a QPT adversary $A$ that has a non-negligible advantage for distinguishing $\hybrid_i$ from $\hybrid_{i+1}$. Based on $A$, we will construct a reduction $R$ to break the security of $\prs$. The reduction $R$ is defined as follows:
\begin{enumerate}
    \item Input: $\sigma^{\otimes t(\secp)}$ where $\sigma$ is sampled from either $\prs(k)$ with a random $k$ or $\Haar_{n(\secp)}$.
    \item For $1\le j\le i$, sample $k_j \gets \bit^\secp$ and run $\prs(k_j)$ $t$ times to get $\rho^{\otimes t(\secp)}_j$.
    \item For $i+2\le j\le q$, sample a $t(\secp)$-state design $\gamma_j$.\footnote{Note that is not efficient for the security reduction to sample Haar random states in each hybrid. Instead of sampling Haar random states, the security reduction uses $t(\lambda)$-state designs. It is known that $t(\lambda)$-state designs can be efficiently generated (in time polynomial in $t(\lambda)$)~\cite{AE07CCC, DCEL09PRA}.}
    \item Run $A \left( \rho^{\otimes t(\secp)}_1, \dots, \rho^{\otimes t(\secp)}_i, \sigma^{\otimes t(\secp)}, \gamma_{i+2}, \dots, \gamma_q \right)$ and output whatever $A$ outputs.
\end{enumerate}
First, the running time $R$ is polynomial in $\secp$. Moreover, $R$ perfectly simulates the view of $A$ and thus has the same distinguishing advantage as that of $A$. However, this contradicts the security of $\prs$.
\end{proof}

\end{lemma}

\noindent We complete the proof of selective security by hybrid arguments. Here, we outline the structure of the hybrids: $\hybrid_1$ is~\Cref{construction:QPRF}. In $\hybrid_{1.i}$ for $i\in\set{0,1,\dots,\secp}$, we replace the output of $\prfs(k_i,\cdot)$ with independent Haar states. The computational indistinguishability between $\hybrid_{1.i}$ and $\hybrid_{1.i+1}$ follows from the selective security of $\prfs$. Finally, in $\hybrid_2$, all the input quantum states of the extractor $\ext$ are now independent Haar states. It remains to show that the resulting output strings are computationally indistinguishable from independent, uniform strings. Fortunately, we observe that we can recycle the proof of strong security of QPRGs in~\Cref{thm:amplification:QPRG} as follows. In $\hybrid_3$, all the independent Haar states are replaced with the output of $\prs$ with \iid uniform seeds. The computational indistinguishability between $\hybrid_{2}$ and $\hybrid_{3}$ follows from~\Cref{lem:poly_copy_PRS}. However, the description of $\hybrid_{3}$ is exactly the same as running the strong QPRG $G^{\xor s}$ defined in~\Cref{thm:amplification:QPRG} on \iid uniform seeds. Hence, the output strings are computationally indistinguishable from independent, uniform strings due to the strong security of $G^{\xor s}$. Formally, consider the following hybrids:

\begin{itemize}
    \item $\hybrid_1:$ In the first hybrid, the adversary receives input-output pairs according to the selective security experiment and~\Cref{construction:QPRF}.
    \begin{enumerate}
        \item Receive $x_1,\dots,x_{q{(\secp)}} \in \bit^{m(\secp)}$ from the adversary.
        \item For $j\in[\secp]$, sample $k_j\gets\bit^\secp$.
        \item For $i\in[q]$, do the following,
        \begin{enumerate}
            \item For $j \in [\secp]$, run $\prfs(k_j, x_i)$ to get $\rho_{k_j, x_i}^{\otimes t(\secp)}$, and run $\ext(\rho_{k_j, x_i}^{\otimes t(\secp)})$ to get $y_{i,j}\in\bit^{\ell(\secp)}$.
            \item Let $y_i = \bigxor_{j = 1}^\secp y_{i,j}\in\bit^{\ell(\secp)}$.
        \end{enumerate}
        \item Output $(x_1, y_1), \dots, (x_q, y_q)$.
    \end{enumerate}

    \item $\hybrid_{1.a}$ for $a \in\set{0, 1, \dots, \secp}:$
    \begin{enumerate}
        \item Receive $x_1,\dots,x_{q{(\secp)}} \in \bit^{m(\secp)}$ from the adversary.
        \item For every $j \in \set{a+1, \dots, \secp}$, sample $k_j \gets \bit^\secp$.
        \item For $i\in[q]$, do the following,
        \begin{enumerate}
            \item For $j \in \set{1, \dots, a}$, sample $\ket{\psi_{i,j}}\gets\Haar(\C^{d(\secp)})$, and run $\ext( \ketbra{\psi_{i,j}}{\psi_{i,j}}^{\otimes t(\secp)} )$ to get $y_{i,j}\in\bit^{\ell(\secp)}$.
            \item For $j \in \set{a+1, \dots, \secp}$, run $\prfs(k_j, x_i)$ to get $\rho_{k_j, x_i}^{\otimes t(\secp)}$, and run $\ext(\rho_{k_j, x_i}^{\otimes t(\secp)})$ to get $y_{i,j}\in\bit^{\ell(\secp)}$.
            \item Let $y_i = \bigxor_{j = 1}^\secp y_{i,j}\in\bit^{\ell(\secp)}$.
        \end{enumerate}
        \item Output $(x_1, y_1), \dots, (x_q, y_q)$.
    \end{enumerate}
    
    \item $\hybrid_2:$ In the second hybrid, all the input of $\ext$ is changed to Haar random states.
    \begin{enumerate}
        \item Receive $x_1,\dots,x_{q{(\secp)}} \in \bit^{m(\secp)}$ from the adversary.
        \item For $i\in[q]$, do the following,
        \begin{enumerate}
            \item For $j \in[\secp]$, sample $\ket{\psi_{i,j}}\gets\Haar(\C^{d(\secp)})$ and run $\ext( \ketbra{\psi_{i,j}}{\psi_{i,j}}^{\otimes t(\secp)} )$ to get $y_{i,j}\in\bit^{\ell(\secp)}$.
            \item Let $y_i = \bigxor_{j = 1}^\secp y_{i,j}\in\bit^{\ell(\secp)}$.
        \end{enumerate}
        \item Output $(x_1, y_1), \dots, (x_q, y_q)$.
    \end{enumerate}

    \item $\hybrid_{3}:$ In the third hybrid, all the input of $\ext$ is changed to the output of an $n(\secp)$-$\prs$.
    \begin{enumerate}
        \item Receive $x_1,\dots,x_{q{(\secp)}} \in \bit^{m(\secp)}$ from the adversary.
        \item For $i\in[q]$, $j\in[\secp]$, sample $k_{i,j}\gets\bit^\secp$.
        \item For $i\in[q]$, do the following,
        \begin{enumerate}
            \item For $j \in[\secp]$, run $\prs(k_{i,j})$ $t$ times to get $\rho_{k_{i,j}}^{\otimes t(\secp)}$, and run $\ext(\rho_{k_{i,j}}^{\otimes t(\secp)})$ to get $y_{i,j}\in\bit^{\ell(\secp)}$.
            \item Let $y_i = \bigxor_{j = 1}^\secp y_{i,j}\in\bit^{\ell(\secp)}$.
        \end{enumerate}
        \item Output $(x_1, y_1), \dots, (x_q, y_q)$.
    \end{enumerate}

    \item $\hybrid_4:$ In the fourth hybrid, each $y_i$ is the output of the $\sqprg$ defined in~\Cref{thm:amplification:QPRG} with $s(\secp)$ set to be $\secp$, where the underlying $\wqprg$ is defined to be the one in~\Cref{construction:wQPRG}.
    \begin{enumerate}
        \item Receive $x_1,\dots,x_{q{(\secp)}} \in \bit^{m(\secp)}$ from the adversary.
        \item For $i\in[q]$, sample $k_i \gets \bit^{\secp^2}$.
        \item For $i\in[q]$, run $y_i \gets \sqprg(k_i)$.
        \item Output $(x_1, y_1), \dots, (x_q, y_q)$.
    \end{enumerate}
    
    \item $\hybrid_5:$ In the last hybrid, the adversary receives independently and uniformly sampled query-answer pairs.
    \begin{enumerate}
        \item Receive $x_1,\dots,x_{q{(\secp)}} \in \bit^{m(\secp)}$ from the adversary.
        \item Sample $y_1,\dots,y_{q{(\secp)}} \gets \bit^{\ell(\secp)}$.
        \item Output $(x_1, y_1), \dots, (x_q, y_q)$.
    \end{enumerate} 
\end{itemize}

\noindent Hybrids $\hybrid_{1}$ and $\hybrid_{1,0}$ are identically distributed. For $a\in[\secp]$, hybrids $\hybrid_{1,a-1}$ and $\hybrid_{1,a}$ are computationally indistinguishable due to the selective security of $\prfs$. Formally, suppose there exists some $a\in[\secp]$ and a QPT adversary $A$ that can distinguish $\hybrid_{1,a-1}$ from $\hybrid_{1,a}$ with non-negligible advantage. We construct a reduction $R$ that breaks the selective security of $\prfs$ as follows:\footnote{Recall that in the definition of selective security, the indices are required to be pairwise distinct.}
\begin{enumerate}
    \item Receive $x_1,\dots,x_{q{(\secp)}} \in \bit^{m(\secp)}$ from the adversary $A$.
    \item For every $j \in \set{a+1, \dots, \secp}$, sample $k_j \gets \bit^\secp$.
    \item Send $x_1, \dots, x_q$ to the challenger and receive $\sigma^{\otimes t(\secp)}_1, \dots, \sigma^{\otimes t(\secp)}_q$, where each $\sigma_i$ is sampled either from $\prfs(k, x_i)$ with the same uniformly random key or $\Haar(\C^d(\secp))$.
    \item For $i\in[q]$, do the following,
        \begin{enumerate}
            \item For $j \in \set{1, \dots, a-1}$, sample a $t(\secp)$-state design $\gamma_{i,j}$ and run $\ext( \gamma_{i,j} )$ to get $y_{i,j}$.
            \item Run $\ext(\sigma^{\otimes t(\secp)}_i)$ to get $y_{i,a}$.
            \item For $j \in \set{a+1, \dots, \secp}$, run $\prfs(k_j, x_i)$ $t$ times to get $\rho_{k_j, x_i}^{\otimes t(\secp)}$, and run $\ext(\rho_{k_j, x_i}^{\otimes t(\secp)})$ to get $y_{i,j}$.
            \item Let $y_i = \bigxor_{j = 1}^\secp y_{i,j}$.
        \end{enumerate}
        \item Run $A((x_1, y_1), \dots, (x_q, y_q))$ and output whatever $A$ outputs.
\end{enumerate}
The reduction $R$ perfectly simulates $A$'s view. Hence, the distinguishing advantage of $R$ is non-negligible, which leads to a contradiction. 
Hybrids $\hybrid_{1, \secp}$ and $\hybrid_{2}$ are identically distributed. 
The computational indistinguishability of hybrids $\hybrid_{2}$ and $\hybrid_3$ follows from the security of $\prs$. In particular, polynomially many samples from $\prs$ with independent, uniform keys are computationally indistinguishable from \iid Haar states as shown in~\Cref{lem:poly_copy_PRS}. Suppose there exists a QPT adversary $A$ that has a non-negligible advantage for distinguishing $\hybrid_{2}$ from $\hybrid_3$. We construct a reduction $R$ that contradicts~\Cref{lem:poly_copy_PRS} as follows:
\begin{enumerate}
    \item Input: $\set{\sigma^{\otimes t(\secp)}_{i,j}}_{i\in[q],j\in[\secp]}$ where all the samples are sampled either from $\prs(k_{i,j})$ with \iid uniform keys or $\Haar_{n(\secp)}$. Note that the number of samples is $q(\secp)\cdot\secp = \poly(\secp)$.
    \item Receive $x_1,\dots,x_{q{(\secp)}} \in \bit^{m(\secp)}$ from the adversary $A$.
    \item For $i\in[q]$, do the following,
        \begin{enumerate}
            \item For $j \in[\secp]$, run $\ext( \sigma^{\otimes t(\secp)}_{i,j} )$ to get $y_{i,j}$.
            \item Let $y_i = \bigxor_{j = 1}^\secp y_{i,j}$.
        \end{enumerate}
    \item Run $A((x_1, y_1), \dots, (x_q, y_q))$ and output whatever $A$ outputs.
\end{enumerate}
Since $R$ runs in polynomial time and has the same distinguishing advantage as that of $A$, this contradicts~\Cref{lem:poly_copy_PRS}.
Hybrids~$\hybrid_3$ and $\hybrid_4$ are syntactically identical. Finally, the computationally indistinguishability between $\hybrid_{4}$ and $\hybrid_{5}$ follows from the strong security of $\sqprg$. To be more precise, similar to classical secure PRGs, polynomially many samples of either the output of a $\sqprg$ with \iid uniform seeds or \iid uniform bitstrings are computationally indistinguishable. The proof is similar to that of~\Cref{lem:poly_copy_PRS}.
\end{proof}

%% file: applications.tex
\section{Applications}
In this section, we present applications based on sQPRGs and selectively secure QPRFs introduced in~\Cref{section:Quantum Pseudorandomess}. One key advantage of using sQPRGs or selectively secure QPRFs as the starting point is that we can build higher-level primitives simply by following the classical construction with a slight modification, and then security will follow from the same reasoning. However, we must address the issue of correctness since sQPRGs and selectively secure QPRFs have only $1 - O(\secp^{-c})$ pseudodeterminism. To resolve the issue, we apply a simple parallel repetition followed by a majority vote to boost correctness at the expense of increased communication complexity and key length.

\subsection{Pseudorandom One-Time Pad (POTP)}
We construct a pseudorandom one-time pad (POTP) scheme with classical communication from sQPRGs. A POTP is an encryption scheme in which the message length is strictly greater than the key length.

\begin{definition}
A \emph{pseudorandom one-time pad (POTP)} for messages of length $\ell(\secp)$ is a triple of QPT algorithms $(\gen, \Enc, \Dec)$ such that the following holds:
\begin{itemize}
\item {\bf{Correctness:}} There exists a negligible function $\veps(\cdot)$ such that for every $\secp\in\N$ and every message $m \in\bit^{\ell(\secp)}$,
\[
\Pr\left[ m = m': \substack{ k \gets \gen(1^\secp), \\
c \gets \Enc(1^\secp,k,m), \\
m' \gets \Dec(1^\secp, k, c)} \right] 
\ge 1 - \veps(\secp).
\]

\item {\bf{Stretch:}} For every $\secp\in\N$, $\ell(\secp) > k(\secp)$, where $k(\secp)$ is the output length of $\gen(1^\secp)$, i.e., the key length.

\item {\bf{Security:}} For any (non-uniform ) QPT adversary $A$, there exists a negligible function $\veps(\cdot)$ such that for every $m_0, m_1 \in\bit^{\ell(\secp)}$ and $\secp\in\N$,
\[
\left| \Pr\left[ A_\secp(c) = 1: \substack{ k \gets \gen(1^\secp),\\ 
c \gets \Enc(1^\secp,k,m_0)} \right] 
- \Pr\left[ A_\secp(c) = 1: \substack{ k \gets \gen(1^\secp),\\
c \gets \Enc(1^\secp,k,m_1)} \right] \right|
\le \veps(\secp).
\]
\end{itemize}
\end{definition}

\noindent Suppose $G_\secp:\bit^\secp\to\bit^{\ell(\secp)}$ is a sQPRG with output length $\ell(\secp) > \secp^2$ and pseudodeterminism $1 - O(\secp^{-c})$ for arbitrary $c > 0$. Consider the following construction:
\begin{construction}[Pseudorandom One-Time Pad (POTP)] \hfill
\label{construction:POTP}
\begin{itemize}
    \item $\gen(1^\secp):$ on input $1^\secp$, outputs a key $k\gets\bit^{\secp^2}$.
    
    \item $\Enc(1^\secp, k, m):$ on input $1^\secp$, a key $k\in\bit^{\secp^2}$ and a message $m \in \bit^{\ell(\secp)}$,
    \begin{enumerate}
        \item Parse $k$ as $k_1||\dots||k_\secp$ such that $k_i \in \bit^\secp$ for every $i\in[\secp]$.
        \item For $i\in[\secp]$, compute $c_i := m \xor G_{\secp}(k_i) \in\bit^{\ell(\secp)}$.
        \item Output $c := c_1||\dots||c_\secp \in \bit^{\secp\ell(\secp)}$.
    \end{enumerate}
    
    \item $\Dec(1^\secp, k, c):$ on input $1^\secp$, a key $k\in\bit^{\secp^2}$ and a ciphertext $c\in\bit^{\secp\ell(\secp)}$,
    \begin{enumerate}
        \item Parse $k$ as $k_1||\dots||k_\secp$ and $c$ as $c_1||\dots||c_\secp$ such that $k_i \in \bit^\secp$ and $c_i \in\bit^{\ell(\secp)}$ for all $i\in[\secp]$.
        \item For $i\in[\secp]$, compute $m_i := c_i\oplus G_\secp(k_i)\in\bit^{\ell(\secp)}$.
        \item Output $m := \MAJ(m_1, \dots, m_\secp)$, where $\MAJ$ denotes the majority function.
    \end{enumerate}
\end{itemize}

\end{construction}

\noindent For the stretch property, the message length of~\Cref{construction:POTP} satisfies $\ell(\secp) > k(\secp) = \secp^2$.

\begin{lemma} \label{lemma:OTP:correctness}
\Cref{construction:POTP} satisfies the correctness property.
\end{lemma}
\begin{proof} 
On input a random key $k = k_1||\dots||k_\secp \in\bit^{\secp^2}$, for every $i\in[\secp]$ we have  $\Pr\left[ k_i \in\cK_\secp \right] \ge 1 - O(\secp^{-c}) > 0.9$ for sufficiently large $\secp$ by the pseudodeterminism of $G$, where $\cK_\secp$ is the set of good seeds of $G$. We denote by $\Good$ the event that at least $0.8\secp$ of the $k_i$'s belong in $\cK_\secp$. For each $k_i$, let $G_E(k_i)$ and $G_D(k_i)$ denote the output of $G$ evaluated by $\Enc$ and $\Dec$, respectively. If $k_i\in\cK_\secp$, then the probability that $G_E(k_i) = G_D(k_i)$ is at least $(1 - O(\secp^{-c}))^2$. Hence, the success probability of the majority vote is at least
\begin{align*}
& \Pr_{k\gets\bit^{\secp^2}}\left[ | \set{i\in[\secp]: G_E(k_i) = G_D(k_i)} | > \secp/2 \right] \\
\ge & \Pr[\Good] \cdot \Pr_{k\gets\bit^{\secp^2}}\left[ | \set{i\in[\secp]: G_E(k_i) = G_D(k_i)} | > \secp/2 \mid \Good \right],
\end{align*}
where the probability is over $k$, $G_E$ and $G_D$.

First, $\Pr[\Good] = 1 - 2^{-\Omega(\secp)}$ by~\Cref{lemma:Chernoff-Hoeffding}. Moreover, conditioned on the event $\Good$ happening, the expected number of $i$'s such that $G_E(k_i) = G_D(k_i)$ is at least $0.8\secp \cdot (1 - O(\secp^{-c}))^2 > 0.7\secp$ for sufficiently large $\secp$. Using~\Cref{lemma:Chernoff-Hoeffding} again, the result of the majority vote is correct with probability at least $1 - \negl(\secp)$.
\end{proof}

\begin{lemma} \label{lemma:OTP:security}
\Cref{construction:POTP} satisfies the security property.
\end{lemma}
\begin{proof}
The security follows from the security of $G$ and a hybrid argument. In particular, consider the following hybrids. Fix $m_0, m_1$.
\begin{itemize}
    \item $\hybrid_0:$ Run $k \gets \gen(1^\secp)$, Run $c\gets\Enc(1^\secp,k,m_0)$. Output $c$.

    \item $\hybrid_{0.a}$ for $a\in\set{0,1,\dots,\secp}:$
    \begin{enumerate}
        \item For $i\in\set{1, \dots, a}$, sample $k_i \gets \bit^{\ell(\secp)}$.
        \item For $i\in\set{a+1, \dots, \secp}$, sample $k_i \gets \bit^{\secp}$.
        \item Output $c = m_0 \xor k_1 || \dots || m_0 \xor k_a 
        || m_0 \xor G_\secp(k_{a+1}) || \dots || m_0 \xor G_\secp(k_\secp)$.
    \end{enumerate}
    
    \item $\hybrid_1:$ For $i\in[\secp]$, sample $c_i \gets \bit^{\ell(\secp)}$.
    Output $c = c_1||\dots||c_\secp$.

    \item $\hybrid_{1.b}$ for $b\in\set{0,1,\dots,\secp}:$
    \begin{enumerate}
        \item For $i\in\set{1,\dots,b}$, sample $k_i \gets \bit^{\secp}$.
        \item For $i\in\set{b+1, \dots, \secp}$, sample $k_i \gets \bit^{\ell(\secp)}$.
        \item Output $c = m_1 \xor G_\secp(k_1) || \dots || m_1 \xor G_\secp(k_b) 
        || m_1 \xor k_{b+1} || \dots || m_1 \xor k_\secp$.
    \end{enumerate}
    
    \item $\hybrid_2:$ Run $k\gets\gen(1^\secp)$. Run $c\gets\Enc(1^\secp,k,m_1)$. Output $c$.
\end{itemize}
First, hybrids $\hybrid_0$ and $\hybrid_{0.0}$ are identically distributed. For $a\in[\secp]$, hybrids $\hybrid_{0.a-1}$ and $\hybrid_{0.a}$ are computational indistinguishabile due to the security of $G$. Specifically, suppose there exist some $a\in[\secp]$ and a QPT adversary $A$ that has a non-negligible advantage for distinguishing $\hybrid_{0.a-1}$ from $\hybrid_{0.a}$. Then consider the following reduction $R$ that breaks the strong security of $G$:
\begin{enumerate}
    \item Input: $y\in\bit^{\ell(\secp)}$ that is either sampled from $G_\secp(k)$ with a uniform seed $k$ or a unifrom $\ell(\secp)$-bit string.
    \item For $i\in\set{1, \dots, a-1}$, sample $k_i \gets \bit^{\ell(\secp)}$.
    \item For $i\in\set{a+1, \dots, \secp}$, sample $k_i \gets \bit^{\secp}$.
    \item Output $c = m_0 \xor k_1 || \dots || m_0 \xor k_{a-1} || m_0 \xor y
        || m_0 \xor G_\secp(k_{a+1}) || \dots || m_0 \xor G_\secp(k_\secp)$.
\end{enumerate}
As $R$ runs in polynomial time and has the same distinguishing advantage as that of $A$, it contradicts the strong security of $G$. Hybrids $\hybrid_{0.\secp}$ and $\hybrid_1$ are identically distributed. Hybrids $\hybrid_1$ and $\hybrid_{1.0}$ are identically distributed. Similarly, for $b\in[\secp]$, hybrids $\hybrid_{1.b-1}$ and $\hybrid_{1.b}$ are computational indistinguishabile due to the security of $G$. Finally, hybrids $\hybrid_{1.\secp}$ and $\hybrid_2$ are identically distributed.
\end{proof}

\subsection{Quantum Commitment with Classical Communication}
Next, we construct a (bit) commitment scheme with classical communication from sQPRGs. We follow the definition in~\cite{AQY21, AGQY22} closely.
\begin{definition}
A bit commitment scheme is given by a pair of (uniform) QPT algorithms $(C,R)$, where $C = \set{C_\secp}_{\secp\in\N}$ is called the \emph{committer} and $R = \set{R_\secp}_{\secp\in\N}$ is called the \emph{receiver}. There are two phases in a commitment scheme: a commit phase and a reveal phase. 
\begin{itemize}
    \item {\bf{Commit phase:}} In the (possibly interactive) commitment phase between $C_\secp$ and $R_\secp$, the committer $C_\secp$ commits to a bit $b$. The communication between $C_{\secp}$ and $R_{\secp}$ is classical.\footnote{Alternately, both the committer and the receiver measure every message they receive in the computational basis.} We denote the execution of the commit phase to be $\sigma_{CR} \gets \mathsf{Commit} \braket{C_\secp(b), R_\secp}$, where $\sigma_{CR}$ is the tensor product of $C_\secp$'s state and $R_\secp$'s state after the commit phase.
    
    \item {\bf{Reveal phase:}} In the reveal phase $C_\secp$ interacts with $R_\secp$ and the output is a trit $\mu\in\set{0,1,\bot}$ indicating the receiver’s output bit or a rejection flag. We denote an execution of the reveal phase where the committer and receiver start with the joint state $\sigma_{CR}$ by $\mu \gets \mathsf{Reveal} \braket{C_\secp(b), R_\secp, \sigma_{CR}}$.
\end{itemize}

\noindent We anticipate the commitment scheme to satisfy the following properties:
\begin{itemize}
    \item {\bf{Correctness:}} We say that a commitment scheme $(C,R)$ satisfies correctness if
    \[
    \Pr\left[b' = b: \substack{ \sigma_{CR} \gets \mathsf{Commit}\braket{C_\secp(b), R_\secp}, \\
    b' \gets \mathsf{Reveal}\braket{C_\secp(b), R_\secp, \sigma_{CR}}} \right] 
    \ge 1 - \veps(\secp),
    \]
    where $\veps(\cdot)$ is a negligible function.
    
    \item {\bf{Computational Hiding:}}
    We say that a commitment scheme $(C,R)$ satisfies computational hiding if for any malicious QPT receiver $\set{R_\secp^*}_{\secp\in\N}$, for any QPT distinguisher $\set{D_\secp}_{\secp\in\N}$, the following holds:
    \[
    \left| \Pr\limits_{(\tau,\sigma_{CR^*}) \gets \mathsf{Commit}\braket{C_\secp(0),R_\secp^*}}[D_\secp(\sigma_{R^*}) = 1]
    - \Pr\limits_{(\tau,\sigma_{CR^*}) \gets \mathsf{Commit}\braket{C_\secp(1),R_\secp^*}}[D_\secp(\sigma_{R^*}) = 1] \right|
    \le \veps(\secp),
    \]
    where $\veps(\cdot)$ is a negligible function and $\tau$ is the transcript in the commitment phase.
    
    \item {\bf{Statistical Binding:}}
     We say that a commitment scheme $(C,R)$ satisfies statistical binding if for any malicious computational unbounded committer $\set{C_\secp^*}_{\secp\in\N}$, the following holds:
    \[
    \Pr \left[ \mathsf{Reveal}\braket{C^*, R_\secp, \sigma_{C^*R}} = 0
    \land \mathsf{Reveal}\braket{C^*, R_\secp, \sigma_{C^*R}} = 1:
    \substack{
        (\tau,\sigma_{C^*R}) \gets \mathsf{Commit}\braket{C^*,R_\secp}
    } \right]
    \le \veps(\secp).
    \]
    where $\veps(\cdot)$ is a negligible function and $\tau$ is the transcript in the commitment phase.
\end{itemize}
\end{definition}

\noindent Suppose $G_\secp:\bit^\secp\to\bit^{\ell(\secp)}$ is a sQPRG with output length $\ell(\secp) = 3\secp$ and pseudodeterminism $1 - O(\secp^{-c})$ for arbitrary $c > 0$. Consider the following construction, which is adapted from Naor's commitment scheme~\cite{Naor89CRTPYO}: 
\begin{construction}[Quantum Bit Commitment with Classical Communication] \hfill
\label{construction:commitment}
\begin{itemize}
    \item {\bf{Commit phase:}}
    \begin{enumerate}
        \item The receiver $R$ samples $r \gets \bit^{3\secp}$ and sends it to the committer $C$.
        \item For $i\in[\secp]$, the committer $C$ samples $k_i \gets \bit^\secp$.
        \item The committer $C$ on input $b\in\bit$, outputs 
        \[
        \Com = \begin{cases}
            G(k_1) || \dots || G(k_\secp) & \text{if } b=0 \\
            G(k_1) \xor r || \dots || G(k_\secp) \xor r & \text{if } b=1. \\
        \end{cases}
        \]
    \end{enumerate}
    
    \item {\bf{Reveal phase:}}
    \begin{enumerate}
        \item The committer $C$ sends the decommitment message $(b, k_1, \dots, k_\secp)$ to the receiver $R$.
        \item The receiver $R$ parses $\Com$ as $y_1||\dots||y_\secp$ where $y_i\in\bit^{3\secp}$ for all $i\in[\secp]$.
        \item For $i\in[\secp]$, the receiver $R$ checks whether $y_i = G(k_i)$ if $b = 0$; checks whether $y_i = G(k_i)\xor r$ if $b = 1$. Let $N\in\set{0,1,\dots,\secp}$ be the number of occurrences where the equality holds
        \item If $N \ge 2\secp/3$, the receiver $R$ outputs $b$; otherwise outputs $\bot$.
    \end{enumerate}
\end{itemize}

\end{construction}

\begin{lemma}
\Cref{construction:commitment} satisfies the correctness property.
\end{lemma}
\begin{proof} 
The proof is similar to that of~\Cref{lemma:OTP:correctness}. In particular, the correctness follows from the pseudodeterminism of $G$ and~\Cref{lemma:Chernoff-Hoeffding}.
\end{proof}

\begin{lemma}
\Cref{construction:commitment} satisfies the computational hiding property.
\end{lemma}
\begin{proof}
The proof is similar to the proof of~\Cref{lemma:OTP:security}. Consider the following hybrids for any fixed $r\in\bit^{3\secp}$.
\begin{itemize}
    \item $\hybrid_0:$ For $i\in[\secp]$, sample $k_i \gets \bit^\secp$. Output $\Com = G(k_1)||\dots||G(k_\secp)$;
    
    \item $\hybrid_{0.a}$ for $a\in\set{0,1,\dots,\secp}:$
    \begin{enumerate}
        \item For $i\in\set{1, \dots, a}$, sample $k_i \gets \bit^{3\secp}$.
        \item For $i\in\set{a+1, \dots, \secp}$, sample $k_i \gets \bit^\secp$.
        \item Output $\Com = k_1 || \dots || k_a || G(k_{a+1}) 
        || \dots || G(k_\secp)$.
    \end{enumerate}
    
    \item $\hybrid_1:$ For $i\in[\secp]$, sample $k_i \gets \bit^{3\secp}$. Output $\Com = k_1 || \dots || k_\secp$.

    \item $\hybrid_2:$ For $i\in[\secp]$, sample $k_i \gets \bit^{3\secp}$. Output $\Com = k_1\xor r || \dots || k_\secp \xor r$.
    
    \item $\hybrid_{2.b}$ for $b \in\set{0,1,\dots,\secp}:$
    \begin{enumerate}
        \item For $i\in\set{1, \dots, b}$, sample $k_i \gets \bit^\secp$.
        \item For $i\in\set{b+1, \dots, \secp}$, sample $k_i \gets \bit^{3\secp}$.
        \item Output $c = G(k_1)\xor r|| \dots || G(k_b) \xor r
        || k_{b+1}\xor r || \dots || k_\secp\xor r$.
    \end{enumerate}
    
    \item $\hybrid_3:$ For $i\in[\secp]$, sample $k_i \gets \bit^\secp$. Output $G(k_1)\xor r||\dots||G(k_\secp)\xor r$;
\end{itemize}
Hybrids $\hybrid_0$ and $\hybrid_{0.0}$ are identically distributed. For $a\in[\secp]$, following the same lines in~\Cref{lemma:OTP:security}, hybrids $\hybrid_{0.a-1}$ and $\hybrid_{0.a}$ are computational indistinguishabile due to the strong security of $G$. Hybrids $\hybrid_{0.\secp}$ and $\hybrid_1$ are identically distributed. Hybrids $\hybrid_1$ and $\hybrid_{2}$ are identically distributed. Hybrids $\hybrid_2$ and $\hybrid_{2.0}$ are identically distributed. Similarly, for $b\in[\secp]$, hybrids $\hybrid_{2.b-1}$ and $\hybrid_{2.b}$ are computational indistinguishabile due to the strong security of $G$. Finally, hybrids $\hybrid_{2.\secp}$ and $\hybrid_3$ identically distributed.
\end{proof}

\begin{lemma}
\Cref{construction:commitment} satisfies the statistical binding property.
\end{lemma}
\begin{proof}
To prove the statistical binding property, we introduce the following definition.
For every $k\in\bit^\secp$, define $F(k):=\argmax_{y\in\bit^{3\secp}} \Pr[G(k) = y]$ (if it is not unique, then we pick the lexicographically first one). Let the set of ``bad randomness'' $\Bad\subseteq\bit^{3\secp}$ be 
\[
\Bad := \set{r\in\bit^{3\secp} \mid \exists k, k'\in\bit^\secp \st F(k) \oplus F(k') = r}.
\]
Then it is easy to see that $\Pr[r\in\Bad:r \gets\bit^{3\secp}] \le 2^\secp\cdot2^\secp/2^{3\secp} = 2^{-\secp}$. Now, the analysis starts to deviate from the proof of the classical case. Classically, if $r \notin \Bad$, then it is impossible for the malicious committer to succeed. However, since now $G$ is \emph{pseudodeterministic}, there is still a chance that $G(k)\xor G(k') = r$ for some $k, k'$ even if $r\notin\Bad$. Fortunately, according to the definition of the set $\Bad$, the XOR of the \emph{most likely} output of $G(k)$ and $G(k')$ for any $k, k'$ must not equal $r$. Below, we will show that the probability that $G(k)\xor G(k') = r$ conditioned on $r\notin\Bad$ is at most $1/2$ for any $k,k'$. 

First, we state a basic fact regarding the inner product of two probability vectors that have distinct mostly likely outcomes.

\begin{claim} \label{claim:prob:inner}
Let $p,q\in\R^n$ be two probability vectors such that $\argmax_{i\in[n]}p_i \neq \argmax_{i\in[n]}q_i$ (if it is not unique, then we pick the lexicographically first one). Then $\sum_{i\in[n]}p_iq_i \le 1/2$.
\end{claim}
\begin{proof}
Without loss of generality, we assume that the coordinates of $p$ are sorted in non-increasing order, i.e., $1\ge p_1\ge p_2\ge \dots \ge p_n \ge 0$. Since $\argmax_{i\in[n]}p_i \neq \argmax_{i\in[n]}q_i$, we have $q_1$ to not be the maximum coordinate.  
\par We claim that there exists $(q'_1,\ldots,q'_n)$ such that the following holds: 
\begin{itemize}
    \item $\forall i \geq 3$, $q'_i=0$,
    \item $\sum_{i \in [n]} p_i q_i \leq \sum_{i \in [n]} p'_i q'_i$.  
    \item $q'_1 \leq q'_2$.  
\end{itemize}
Suppose we instantiate $q'_1=q_1$ and $q'_2=\sum_{i \geq 2} q_i$ then the above three bullet points hold. 
\par Now, we have the following: 
\begin{eqnarray*}
\sum_{i \in [n]} p_i q'_i & = & p_1 q'_1 + p_2 q'_2.
\end{eqnarray*}
Since $q'_1 \leq q'_2$, $q'_1 + q'_2 = 1$ and $p_1 \ge p_2$, the above expression is maximized when $q'_1=q'_2 = \frac{1}{2}$. Thus, $\sum_{i \in [n]} p_i q'_i \leq \frac{1}{2}(p_1 + p_2)\leq \frac{1}{2}$. 
This further implies that $\sum_{i \in [n]} p_i q_i \leq \frac{1}{2}$. 
\end{proof}

\begin{claim} \label{claim:collision}
For every $r\notin\Bad$ and every $k, k'\in\bit^\secp$, $\Pr[G(k) \xor G(k') = r]\le 1/2$, where the probability is over the randomness of $G$.
\end{claim}
\begin{proof}
Fix $k,k'$ and $r\notin\Bad$. First, recall that $r\notin\Bad$ means $F(k) \oplus F(k') \neq r$. The probability can be written as
\begin{align*}
\Pr[G(k)\xor G(k') = r] = \sum_{z\in\bit^{3\secp}} \Pr[G(k) = z]\Pr[G(k') = z\xor r].
\end{align*}
Now, we define the probability vectors $u, u'\in\R^{2^{3\secp}}$ for the random variables $G(k), G(k')$ respectively. More precisely, the coordinate of $u$ is defined to be $u_y := \Pr[G(k) = y]$; $u'$ is defined similarly. We use the above notation to rewrite the quantity as follows.
\begin{align*}
& \sum_{y\in\bit^{3\secp}} \Pr[G(k) = y] \Pr[G(k') = y\xor r] 
= \sum_{y\in\bit^{3\secp}} u_y \cdot u'_{y \xor r}.
\end{align*}
Then we set $p$ and $q$ in~\Cref{claim:prob:inner} to be the vertors that satisfy $p_y = u_y$ and $q_y = u_{y\xor r}$ for all $y\in\bit^{3\secp}$. Let $y_{\max}$, $y'_{\max}$ be the most likely outcome of $G(k)$, $G(k')$ respectively. Given that $r \notin \Bad$, we have $y_{\max}\xor y'_{\max} \neq r$. Hence, $u$ and $u'$ satisfy the condition in~\Cref{claim:prob:inner}. Finally, by~\Cref{claim:prob:inner}, we can conclude that $\Pr[G(k)\xor G(k') = r] \le 1/2$.
\end{proof}

\noindent To prove statistical binding, we have
\begin{align*}
& \Pr \left[ \mathsf{Reveal}\braket{C^*, R_\secp, \sigma_{C^*R}} = 0
    \land \mathsf{Reveal}\braket{C^*, R_\secp, \sigma_{C^*R}} = 1:
    \substack{
        (\tau,\sigma_{C^*R}) \gets \mathsf{Commit}\braket{C^*,R_\secp}
    } \right] \\
=  & \Ex\limits_{r\gets\bit^{3\secp}} \left[ \max_{k,k'\in\bit^{\secp^2}} \Pr\left[ G(k_1)||\dots||G(k_\secp) =  G(k'_1)\xor r||\dots||G(k'_\secp)\xor r \right] \right] \\
\le & \Pr\limits_{r\gets\bit^{3\secp}}[r\in \Bad] + \\
& \Ex\limits_{r\gets\bit^{3\secp}} \left[ \max_{k,k'\in\bit^{\secp^2}} \Pr\left[ G(k_1)||\dots||G(k_\secp) =  G(k'_1)\xor r||\dots||G(k'_\secp) \xor r \right] \mid r \notin \Bad \right].
\end{align*}

\noindent The first term $\Pr[r\in\Bad]$ is at most $2^{-\secp}$ as we shown. Let $\xi(n,p)$ be the probability that there are at least $2n/3$ heads when independently tossing a coin $n$ times, where the coin satisfies that $\Pr[\mathsf{Head}] = p$. Then, the second term can be written as
\begin{align*}
& \Ex\limits_{r\gets\bit^{3\secp}} \left[ \max_{k,k'\in\bit^{\secp^2}} \Pr\left[ G(k_1)||\dots||G(k_\secp) =  G(k'_1)\xor r||\dots||G(k'_\secp) \xor r \right] \mid r \notin \Bad \right] \\
& = \Ex\limits_{r\gets\bit^{3\secp}} \left[ \xi \left(\secp, \max_{k,k'\in\bit^\secp} \Pr[G(k)\xor G(k') = r] \right) \mid r\notin\Bad \right] \\
& \le \xi\left(\secp, \frac{1}{2} \right)
= 2^{-\Omega(\secp)},
\end{align*}
where the inequality follows from~\Cref{claim:collision}; the last equality follows from~\Cref{lemma:Chernoff-Hoeffding}.
\end{proof}


\subsection{Non-Adaptive CPA-Secure Quantum Private-Key Encryption with Classical Ciphertexts}
Finally, we construct a non-adaptive CPA-Secure private-key encryption with classical ciphertexts from selectively secure QPRFs.

\begin{definition}[Non-Adaptive CPA-Secure Quantum Private-key Encryption]
We say that a tuple of QPT algorithms $(\gen, \Enc, \Dec)$ is a \emph{non-adaptive CPA-secure quantum private-key encryption scheme} if the following holds:
\begin{itemize}
    \item {\bf{Correctness:}} There exists a negligible function $\veps(\cdot)$ such that for every $\secp\in\N$ and every message $m$,
    \[
    \Pr_{k \gets \bit^\secp} \left[ \Dec(1^\secp, k, \Enc(1^\secp, k, m)) = m \right] \ge 1 - \veps(\secp).
    \]
    \item {\bf{Non-adaptive CPA security:}} For every polynomial $q(\cdot)$, any (non-uniform) QPT adversary $A$, there exists a negligible function $\veps(\cdot)$ such that for all $\secp\in\N$, the adversary $A$ has at most $\veps(\secp)$ advantage in the following experiment:
    
    \begin{enumerate}
    \item The challenger generates a key $k$ by running $\gen(1^\secp)$ and a uniform bit $b\in\bit$.
    \item The adversary $A$ is given input $1^\secp$.
    \item The adversary $A$ chooses messages $(m^0_1, m^1_1) \dots, (m^0_q, m^1_q)$ and sends them to the challenger.
    \item The challenger sends $\Enc(k,m^b_1), \dots, \Enc(k,m^b_q)$ to the adversary $A$.
    \item The adversary $A$ outputs a bit $b'\in\bit$.
    \item The challenger output $1$ if $b' = b$, and $0$ otherwise.
\end{enumerate}
\end{itemize}

\end{definition}

Suppose $F:\bit^\secp \times \bit^{m(\secp)}\to\bit^{\ell(\secp)}$ is a selectively secure QPRF with $m(\secp) = \omega(\log\secp)$ and pseudodeterminism $1 - O(\secp^{-c})$ for arbitrary $c > 0$. In the classical case, selectively secure pseudorandom functions imply the existence of non-adaptive CPA-secure private-key encryption schemes. With a slight modification, we have the following construction.
\begin{construction}[Non-Adaptive CPA-Secure Quantum Private-Key Encryption Scheme] \hfill
\label{construction:SKE}
\begin{enumerate}
    \item $\gen(1^\secp):$ on input $1^\secp$, output $k\gets\bit^{\secp^2}$.
    
    \item $\Enc(1^\secp, k, m):$ on input a key $k\in\bit^{\secp^2}$ and a message $m \in\bit^{\ell(\secp)}$,
    \begin{itemize}
        \item Parse $k$ as $k_1||\dots||k_\secp$ such that $k_i \in \bit^\secp$ for every $i\in[\secp]$.
        \item Choose a uniform string $r\gets\bit^{m(\secp)}$.
        \item For $i\in[\secp]$, compute $F(k_i, r)$.
        \item Output $c = (r, m\xor F(k_1, r), \dots, m\xor F(k_\secp, r) )$.
    \end{itemize}
    
    \item $\Dec(1^\secp, k, c):$ on input a key $k\in\bit^{\secp^2}$ and a ciphertext $c = (r,c_1,\dots,c_\secp)\in\bit^{m(\secp) + \secp\ell(\secp)}$,
    \begin{itemize}
        \item Parse $k$ as $k_1||\dots||k_\secp$ such that $k_i \in \bit^\secp$ for every $i\in[\secp]$.
        \item For $i\in[\secp]$, compute $m_i := c_i \oplus F(k_i, r) \in \bit^{\ell(\secp)}$.
        \item Output $m := \MAJ(m_1, \dots, m_\secp)$.
    \end{itemize}
\end{enumerate}
\end{construction}

\begin{lemma}
\Cref{construction:SKE} satisfies the correctness property.
\end{lemma}
\begin{proof}
The proof is similar to that of~\Cref{lemma:OTP:correctness}. The correctness follows from the pseudodeterminism of $F$ and~\Cref{lemma:Chernoff-Hoeffding}.
\end{proof}

\noindent The following lemma follows the proof of Lemma~7.3 in~\cite{AQY21} closely.
\begin{lemma}
\Cref{construction:SKE} satisfies non-adaptive CPA security.
\end{lemma}
\begin{proof}
We finish the proof with a hybrid argument. Consider the following hybrids.
\begin{itemize}
    \item $\hybrid_1:$ The adversary receives $\left( \Enc(k,m^b_1), \dots, \Enc(k,m^b_q) \right)$, which by definition is 
    \[ \left( 
    (r_1, m^b_1 \xor F(k_1, r_1), \dots, m^b_1 \xor F(k_\secp, r_1)), \dots, (r_q, m^b_q \xor F(k_1, r_q), \dots, m^b_q \xor F(k_\secp, r_q))
    \right) \]
    where $r_1, \dots, r_q$ are independently and uniformly chosen.

    \item $\hybrid_{1.i}$ for $i\in\set{0, 1, \dots, \secp}:$ The adversary receives
    \begin{align*}
    \Bigg( 
    \left( r_1, m^b_1 \xor R_1(r_1), \dots, m^b_1 \xor R_i(r_1), m^b_1 \xor F(k_{i+1}, r_1)), \dots, m^b_1 \xor F(k_\secp, r_1) \right), \\
    \dots, 
    \left( r_q, m^b_q \xor R_1(r_q), \dots, m^b_q \xor R_i(r_q), m^b_q \xor F(k_{i+1}, r_q)), \dots, m^b_q \xor F(k_\secp, r_q) \right)
    \Bigg)
    \end{align*}
    where $r_1, \dots, r_q$ are independently, uniformly chosen and $R_1(\cdot), \dots, R_i(\cdot)$ are independent random functions.

    \item $\hybrid_{2}:$ The adversary receives
    \begin{align*}
    \Bigg( 
    \left( r_1, m^b_1 \xor R_1(r_1), \dots, m^b_1 \xor R_\secp(r_1) \right), \dots, 
    \left( r_q, m^b_q \xor R_1(r_q), \dots, m^b_q \xor R_\secp(r_q) \right)
    \Bigg)
    \end{align*}
    where $r_1, \dots, r_q$ are independently, uniformly chosen and $R_1(\cdot), \dots, R_\secp(\cdot)$ are independent random functions.
    
    \item $\hybrid_3:$ Instead of sampling $r_1, \dots, r_q$ independently, they are sampled uniformly at random conditioned on them all being distinct. The adversary receives
    \begin{align*}
    \Bigg( 
    \left( r_1, m^b_1 \xor R_1(r_1), \dots, m^b_1 \xor R_\secp(r_1) \right), \dots, 
    \left( r_q, m^b_q \xor R_1(r_q), \dots, m^b_q \xor R_\secp(r_q) \right)
    \Bigg)
    \end{align*}
    where $R_1(\cdot), \dots, R_\secp(\cdot)$ are independent random functions.
\end{itemize}

\noindent Hybrids $\hybrid_1$ and $\hybrid_{1.0}$ are identically distributed. For $i\in\set{0,1,\dots,\secp-1}$, hybrids $\hybrid_{1.i}$ and $\hybrid_{1.i+1}$ are computational indistinguishabile from the selective security of $F$. In particular, suppose there exist some $i\in[\secp]$ and a QPT adversary $A$ such that the difference between $A$'s winning probabilities in $\hybrid_{1.i}$ and $\hybrid_{1.i+1}$ is non-negligible. Consider the following reduction $R$ that breaks the selective security of the underlying QPRF $F$.

\begin{enumerate}
    \item Receive $(m^0_1, m^1_1) \dots, (m^0_q, m^1_q)$ from $A$.
    \item Sample $x_1,x_2,\dots,x_q \gets \bit^{m(\secp)}$.
    \item Query the oracle on $x_1,x_2,\dots,x_q$ and obtain $y_1,y_2,\dots,y_q$, where $y_i$'s are either sampled from $F(k,x_i)$ with a uniform key $k$ or \iid uniform $\ell(\secp)$-bit strings.
    \item For $j\in\set{1,2,\dots,q}$, do the following
    \begin{enumerate}
        \item Sample independent random functions $R_1(\cdot), \dots, R_i(\cdot)$ and compute $R_1(x_j),$ $R_2(x_j),$ $\dots, R_i(x_j)$.\footnote{The reduction $R$ uses lazy evaluation to simulate each random function instead of sampling the whole function table.}
        \item Sample $k_{i+2},\dots, k_\secp \gets \bit^\secp$ and compute $F(k_{i+2},x_j),\dots,F(k_\secp,x_j)$.
    \end{enumerate}
    \item Sample a uniform bit $b\in\bit$.
    \item Send 
    \begin{align*}
    & ( x_1, m^b_1 \xor R_1(x_1), \dots, m^b_1 \xor R_i(x_1), m^b_1 \xor y_1, m^b_1 \xor F(k_{i+2},x_1), \dots, F(k_\secp,x_1) ), \\
    & ( x_2, m^b_2 \xor R_1(x_2), \dots, m^b_2 \xor R_i(x_2), m^b_2 \xor y_2, m^b_2 \xor F(k_{i+2},x_2), \dots, F(k_\secp,x_2) ), \\
    & \quad \vdots \\
    & ( x_q, m^b_q \xor R_1(x_q), \dots, m^b_q \xor R_i(x_q), m^b_q \xor y_q, m^b_q \xor F(k_{i+2},x_q), \dots, F(k_\secp,x_q) )
    \end{align*}
    to $A$ and get the output $b'$.
    \item If $b = b'$, then output $1$. Otherwise, output $0$.
\end{enumerate}
If $y_i$'s are the output of the QPRF $F$, then the reduction $R$ perfectly simulates $A$'s view in $\hybrid_i$. On the other hand, if $y_i$'s are \iid uniform bitstrings, then the reduction $R$ perfectly simulates $A$'s view in $\hybrid_{i+1}$. Hence, the distinguishing advantage of $R$ is equivalent to the difference between $A$'s winning probabilities in $\hybrid_{1.i}$ and $\hybrid_{1.i+1}$. However, this contradicts the selective security of $F$.
Hybrids $\hybrid_{1.\secp}$ and $\hybrid_2$ are identically distributed. The statistical distance between hybrids $\hybrid_2$ and $\hybrid_3$ is $O(q^2/2^m) = \negl(\secp)$ from a similar calculation of Lemma~7.3 in~\cite{AQY21}. Finally, the advantage of $A$ in $\hybrid_3$ is $0$ since all the messages are independently one-time padded. 

\end{proof}

%% file: ProofAmplification.tex
\section{Full Proof of~\Cref{thm:amplification:QPRG}}
\label{app:amplification:PRG}
In this section, we aim to complete the proof of~\Cref{thm:amplification:QPRG}. The proof is essentially the same as that of~\cite{DIJK09TCC}. In~\cite{DIJK09TCC}, \Cref{thm:amplification:PRG} is proven by a direct product theorem in~\cite{IJK08JOC}. This is partly because they consider more general settings, e.g., interactive primitives and direct products with thresholds. For our purpose, we can use the hardness amplification of weakly verifiable puzzles by Canetti, Halevi, and Steiner~\cite{CHS05TCC}. Specifically, Lemma~14 in~\cite{DIJK09TCC} can be proven by techniques in~\cite{CHS05TCC}. We note that Radian and Sattath~\cite{RS19} observed the result in~\cite{CHS05TCC} can be extended to the post-quantum setting without modifying the proof; Morimae and Yamakawa~\cite{MY22onewayness} further generalized the result in~\cite{CHS05TCC} to the setting where the puzzle and solution are quantum under certain conditions.

We first recall the definition of weakly verifiable puzzles in~\cite{CHS05TCC} and extend it to our setting where the puzzle generator, verifier and solver are QPTs while the puzzle and solution remain classical.

\begin{definition}[Weakly verifiable puzzles]
A weakly verifiable puzzle is a pair of QPT algorithms $\Pi = (\gen,\ver)$ that satisfies the following:
\begin{enumerate}
    \item $\gen(1^\secp)\to (P,C):$ on input $1^\secp$, outputs a classical puzzle $P$ along with a classical (secret) bitstring $C$ for verification.
    \item $\ver(P,S,C)\to\top/\bot:$ on input a classical puzzle $P$, a classical solution $S$ and a classical bitstring $C$,  outputs the symbol $\top$ if the verification is successful or $\bot$ if it fails.
\end{enumerate}
\end{definition}

\begin{theorem}[{\cite[Theorem~1]{CHS05TCC}},{\cite[Theorem~20]{RS19}}] \label{theorem;ParallelRepeatition}
Let $\veps:\N \to [0,1]$ be an efficiently computable function, let $s:\N \rightarrow \N$ be efficiently computable and polynomially bounded, and let $\Pi = (\gen,\ver)$ be a weakly verifiable puzzle system. If $\Pi$ is $(1-\veps)$-hard against QPT adversaries, then $\Pi^s$, the $s$-fold repetition of $\Pi$, is $(1-\veps^s)$-hard against QPT adversaries.
\end{theorem}

Next, we note that Yao's next-bit unpredictability lemma~~\cite{Yao82FOCS} can be extended to QPRGs.  The proof is almost the same as the classical case except that QPRGs are pseudodeterministic -- which means that the first $i$ bits of the QPRG's output and its next bit need to be well-defined. To address this issue, we sample $G(k)$ once and define the first $i$ bits and the next bit accordingly. For completeness, we present the proof below.

\begin{lemma}[Next-bit unpredictability lemma for QPRGs]
\label{lemma:next-bit}
If a QPT algorithm $G:\bit^\secp\to\bit^{\ell(\secp)}$ is $(1-\delta)$-pseudorandom, then for any QPT algorithm $A$ and any $i\in\set{0,1,\dots,\ell(\secp)-1}$,
\[
\Pr\left[ b' = y_{i+1}: 
\substack{ k \gets \bit^\secp, \\ 
y \gets G(k), \\
b' \gets A\left( y_{[1:i]} \right) }
\right] 
\le \frac{1}{2} + \delta,
\]
where $y_{[1:i]}$ denotes the first $i$ bits of $y$, and $y_{i+1}$ denotes the $(i+1)$-th bit of $y$.

Let $G:\bit^\secp\to\bit^{\ell(\secp)}$ be a QPT algorithm such that for any QPT algorithm $A$ and any $i\in\set{0,1,\dots,\ell(\secp)-1}$,
\[
\Pr\left[ b' = y_{i+1}: 
\substack{ k \gets \bit^\secp, \\ 
y \gets G(k), \\
b' \gets A\left( y_{[1:i]} \right) }
\right] 
\le \frac{1}{2} + \delta.
\]
Then $G$ is $(1 - \delta\ell)$-pseudorandom.
\end{lemma}
\begin{proof}
\noindent (\textit{Pseudorandomness implies next-bit unpredictability.}) For the sake of contradiction, suppose there exists some $i\in\set{0,1,\dots,\ell(\secp)-1}$ and a QPT algorithm $A$ such that 
\[
\Pr\left[ b' = y_{i+1}: 
\substack{ k \gets \bit^\secp, \\ 
y \gets G(k), \\
b' \gets A\left( y_{[1:i]} \right) }
\right] 
> \frac{1}{2} + \delta.
\]
We construct the following reduction $R$ the breaks the security of $G$: on input $y\in\bit^{\ell(\secp)}$, do the following,
\begin{enumerate}
    \item Run $A(y_{[1:n]})$.
    \item Receive a bit $b'$ from $A$.
    \item Output $1$ if $b' = y_{i+1}$, and $0$ otherwise.
\end{enumerate}
When the input string $y$ is sampled from $G(k)$ with a random seed $k$, then the reduction perfectly simulates the view of $A$. On the other hand, if $y$ is a uniform string, then the probability of $b' = y_{i+1}$ is $1/2$ since the information of $y_{i+1}$ is never revealed. Hence, we have
\[
\Pr\left[R(y) = 1: \substack{ k\gets\bit^\secp, \\ y \gets G(k)} \right]
- \Pr[R(y) = 1: y \gets \bit^{\ell(\secp)} ]
\ge \left( \frac{1}{2} + \delta \right) - \frac{1}{2} 
= \delta.
\]
However, this contradicts the assumption that $G$ is $(1-\delta)$-pseudorandom. \\

\noindent (\textit{Next-bit unpredictability implies pseudorandomness.}) For the sake of contradiction, suppose there exists a QPT algorithm $A$ that breaks the pseudorandomness of $G$, i.e.,
\[
\left| \Pr[A(y) = 1: k\gets\bit^\secp, y \gets G(k) ] - \Pr[A(y) = 1: y \gets \bit^{\ell(\secp)} ] \right| > \delta\ell.
\]
Consider the following hybrids:
\begin{itemize}
    \item $\hybrid_0:$ $u \gets \bit^{\ell(\secp)}$; Output $u$.
    \item $\hybrid_i$ for $i\in\set{1,\dots,\ell(\secp) - 1}:$ $k\gets\bit^\secp$; $y \gets G(k)$; $u \gets \bit^{\ell(\secp)}$; Output $y_{[1:i]}||u_{[i+1:\ell]}$.
    \item $\hybrid_{\ell}:$ $k\gets\bit^\secp$; $y \gets G(k)$; Output $y$.
\end{itemize}
Hence, by hybrid argument, there exists some $i\in[\ell]$ such that
\[
\left| \Pr_{\hybrid_{i}}[A(y) = 1] - \Pr_{\hybrid_{i-1}}[A(y) = 1] \right| > \frac{\delta\ell}{\ell} = \delta.
\]
Without loss of generality, we can assume that $\Pr_{\hybrid_{i}}[A(y) = 1] - \Pr_{\hybrid_{i-1}}[A(y) = 1] > \delta$. Now, we construct a reduction $R$ that takes as input the first $i-1$ bits of $G(k)$ and runs $A$ to predict the $i$-th bit of $G(k)$. Consider the following reduction $R$: on input $y = y_1||\dots||y_{i-1}$, do the following,
\begin{enumerate}
    \item Sample uniform bits $u_i,\dots,u_\ell \in\bit$.
    \item Run $A(y_1||\dots||y_{i-1}, u_i,\dots,u_\ell)$.
    \item Receive a bit $b'$ from $A$.
    \item Output $u_i$ if $b' = 1$, and $u_i\xor 1$ otherwise.
\end{enumerate}
Note that the first two steps perfectly simulate $\hybrid_{i-1}$ for $A$, thus $\Pr[b' = 1] = \Pr_{\hybrid_{i-1}}[A(y) = 1]$. Moreover, conditioned on $u_i = y_i$, the reduction perfectly simulates $\hybrid_{i}$ for $A$, thus $\Pr[b' = 1 \mid u_i = y_i] = \Pr_{\hybrid_{i}}[A(y) = 1]$. Now, since $u_i$ is chosen uniformly at random, we have $\Pr[b' = 1] = \frac{1}{2}\Pr[b' = 1 \mid u_i = y_i] + \frac{1}{2}\Pr[b' = 1 \mid u_i \neq y_i]$, or equivalently $\Pr[b' = 1 \mid u_i \neq y_i] = 2\Pr[b' = 1] - \Pr[b' = 1 \mid u_i = y_i] = 2\Pr_{\hybrid_{i-1}}[A(y) = 1] - \Pr_{\hybrid_{i}}[A(y) = 1]$.

\noindent Then the success probability of $R$ is given by
\begin{align*}
& \Pr[u_i = y_i \land b' = 1] + \Pr[u_i \neq y_i \land b' = 0] \\
& = \Pr[u_i = y_i] \Pr[b' = 1 \mid u_i = y_i] + \Pr[u_i \neq y_i] \Pr[b' = 0 \mid u_i \neq y_i] \\
& = \Pr[u_i = y_i] \Pr[b' = 1 \mid u_i = y_i] + \Pr[u_i \neq y_i] (1 - \Pr[b' = 1 \mid u_i \neq y_i]) \\
& = \frac{1}{2}\Pr_{\hybrid_{i}}[A(y) = 1] + \frac{1}{2} \left( 1 - 2\Pr_{\hybrid_{i-1}}[A(y) = 1] + \Pr_{\hybrid_{i}}[A(y) = 1] \right) \\
& = \frac{1}{2} + \left( \Pr_{\hybrid_{i}}[A(y) = 1] - \Pr_{\hybrid_{i-1}}[A(y) = 1] \right) \\
& \ge \frac{1}{2} + \delta.
\end{align*}
However, this contradicts the next-bit unpredictability of $G$.
\end{proof}

Building upon the argument in~\cite{DIJK09TCC}, we can interpret $y_{[1:i]}$ as the \emph{puzzle} and the next bit $y_{i+1}$ as the \emph{solution}. Then we define the $(2s)$-fold puzzle as $y^1_{[1:i]} ||\dots|| y^{2s}_{[1:i]}$, where $k_1,\dots, k_{2s}$ are independently, uniformly chosen and $y^1 \gets G(k_1), \dots, y^{2s} \gets G(k_{2s})$. The solution will be $y^1_{i+1} ||\dots|| y^{2s}_{i+1}$. Here, we extend Lemma~14 in~\cite{DIJK09TCC} that states the hardness amplification for the above puzzle to QPRGs.

\begin{lemma}[Bit-wise direct product lemma for QPRGs]
\label{lemma:Bit-wise direct product lemma}
Let $G:\bit^\secp\to\bit^{\ell(\secp)}$ be a QPRG such that for any QPT algorithm $A$, we have for all $i\in\set{0,1,\dots,\ell(\secp)-1}$,
\[
\Pr\left[ b' = y_{i+1}: 
\substack{ k \gets \bit^\secp, \\ 
y \gets G(k), \\
b' \gets A\left( y_{[1:i]} \right) }
\right]
\le \frac{1}{2} + \delta,
\]
where $\delta(\secp) \le 0.49+o(1)$. Then for any QPT algorithm $A'$ we have for all $i\in\set{0,1,\dots,\ell(\secp)-1}$,
\[
\Pr\left[ \bigwedge_{j = 1}^{2s} \left( b'_j = y^j_{i+1} \right): 
\substack{
    k_1, \dots, k_{2s} \gets \bit^\secp, \\
    y^1 \gets G(k_1), \dots, y^{2s} \gets G(k_{2s}), \\
    (b'_1, \dots, b'_{2s}) \gets A' \left( y^1_{[1:i]}, \dots, y^{2s}_{[1:i]} \right) 
}
\right] 
\le \veps,
\]
where $\veps = e^{-\Omega(s)}$.
\end{lemma}
\begin{proof}
Following the arguments in~\cite{DIJK09TCC}, we model the above problem as a weakly verifiable puzzle. Specifically, the puzzle generator $\gen$ is defined as follows: sample $k_1, \dots, k_{2s} \gets \bit^\secp$ and then run $y^1 \gets G(k_1), \dots, y^{2s} \gets G(k_{2s})$. The puzzle $P$ is $\set{ y^1_{[1:i]}, \dots, y^{2s}_{[1:i]} }$. The classical (secret) bitstring $C$ is $\set{y^1_{i+1}, \dots, y^{2s}_{i+1}}$. The solution $S$ is of the form $\set{b'_1, \dots, b'_{2s} }$. The puzzle verifier $\ver$ takes as input $(C, P, S)$ and outputs $\top$ if and only if it satisfies $\bigwedge_{j = 1}^{2s} \left( b'_j = y^j_{i+1} \right)$. Finally, we apply~\Cref{theorem;ParallelRepeatition} on the above puzzle and obtain $\veps = O\left( \left( \frac{1}{2}+\delta \right)^{2s} \right) = e^{-\Omega(s)}$.
\end{proof}

\noindent Moreover, Lemma~15 in~\cite{DIJK09TCC} also can be generalized for QPRGs.

\begin{lemma}[Direct product theorem implies xor lemma for QPRGs]
\label{lemma:Direct product theorem implies xor lemma}
Let $G:\bit^\secp\to\bit^{\ell(\secp)}$ be a QPRG such that for any QPT $A$, we have for all $i\in\set{0,1,\dots,\ell(\secp)-1}$,
\[
\Pr\left[ b' = y_{i+1}: 
\substack{ k \gets \bit^\secp, \\ 
y \gets G(k), \\
b' \gets A\left( y_{[1:i]} \right) }
\right]
\le \frac{1}{2} + \delta,
\]
where $\delta(\secp) \le 0.49+o(1)$. Then for any QPT algorithm $A'$ we have for all $i\in\set{0,1,\dots,\ell(\secp)-1}$,
\[
\Pr\left[ b' = \bigoplus_{j=1}^s y^j_{i+1}: 
\substack{
    k_1, \dots, k_s \gets \bit^\secp, \\
    y^1 \gets G(k_1), \dots, y^s \gets G(k_s), \\
    b' \gets A' \left( y^1_{[1:i]}\xor \dots\xor y^s_{[1:i]} \right) 
}
\right]
\le \frac{1}{2} + \veps,
\]
where $\veps = e^{-\Omega(s)}$.
\end{lemma}
\begin{proof}
The proof is essentially the same as that of~\cite{DIJK09TCC}. For completeness, we sketch the proof below. Suppose there exists a QPT adversary $A'$ such that 
\[
\Pr\left[ b' = \bigoplus_{j=1}^s y^j_{i+1}: 
\substack{
    k_1, \dots, k_s \gets \bit^\secp, \\
    y^1 \gets G(k_1), \dots, y^s \gets G(k_s), \\
    b' \gets A' \left( y^1_{[1:i]}\xor \dots\xor y^s_{[1:i]} \right) 
}
\right]
> \frac{1}{2} + \veps.
\]
holds for some $i\in\set{0,1,\dots,\ell(\secp)-1}$ and $\veps \notin e^{-\Omega(s)}$.
Then we will construct a reduction $A''$ that on input a random string $r\in\bit^{2s}$ and $y^1_{[1:i]}, \dots, y^{2s}_{[1:i]} \in\bit^i$, output the inner product of $r$ and $y^1_{i+1} || \dots || y^{2s}_{i+1}$, where $y^1\gets G(k_1), \dots, y^{2s}\gets G(k_{2s})$ and $k_1, \dots, k_{2s}$ are independent, uniform seeds. The description of $A''$ is the following: on input $r\in\bit^{2s}$ and $y^1_{[1:i]}, \dots, y^{2s}_{[1:i]} \in\bit^i$, do the following,
\begin{enumerate}
    \item If the number of $1$'s in $r$ is not $s$, then output a random bit $b$.
    \item Otherwise, let $z_1,\dots,z_s \in [2s]$ be the indices such that $r_{z_j} = 1$ for all $j \in [s]$.
    \item Run $A'(y^{z_1}_{[1:i]}\xor \dots \xor y^{z_s}_{[1:i]})$.
    \item Output whatever $A'$ outputs.
\end{enumerate}
First, note that when $r$ has exactly $s$ $1$'s, the inner product satisfies $\braket{r, y^1_{i+1} || \dots || y^{2s}_{i+1} } = y^{z_1}_{i+1} \xor \dots \xor y^{z_s}_{i+1}$. Moreover, the probability that $r$ has exactly $s$ $1$'s with probability $\Theta(1/\sqrt{s})$. This implies that $A''$ computes the above inner product with probability at least $1/2 + \veps'$, where $\veps' = \Theta(\veps/\sqrt{s})$. By averaging, with probability at least $\veps'/2$ the tuples $(k_1,y^1), \dots, (k_{2s},y^{2s})$ are ``good'' such that $A''$ computes the inner product with a randomly chosen $r$ with probability at least $1/2 + \veps'/2$. Now, using the Goldreich-Levin Theorem,\footnote{The adversary $A'''$ receives,  as non-uniform advice, multiple copies of the non-uniform advice of $A''$ and thus, $A'''$ can execute $A''$ many times in the Goldreich-Levin reduction. Alternately, we can use the quantum Goldreich-Levin theorem~\cite{AC02}.} we can construct $A'''$ which for every good $(k_1,y^1), \dots, (k_{2s},y^{2s})$, computes $y^1_{i+1} || \dots || y^{2s}_{i+1}$ with probability at least $\Theta((\veps')^2/s)$. This implies that $A'''$ computes $y^1_{i+1} || \dots || y^{2s}_{i+1}$ with probability at least $\Omega(\veps^3/s^{5/2})$ which contradicts~\Cref{lemma:Bit-wise direct product lemma}.

\end{proof}

\noindent Finally, we complete the proof of strong security in~\Cref{thm:amplification:QPRG}, we restate the theorem for convenience.
\begin{theorem}[\Cref{thm:amplification:QPRG}]
Let $G:\bit^\secp \to \bit^{\ell(\secp)}$ be a wQPRG that has pseudodeterminism $1 - O(\secp^{-c})$ and pseudorandomness $1 - \delta$ such that $c>1$, $\delta(\secp) \le 0.49+o(1)$ and $\ell(\secp) > s(\secp)\cdot \secp$, where $s(\secp) = \Theta(\secp)$. Define the QPT algorithm $G^{\xor s}: \bit^{s(\secp)\secp} \to \bit^{\ell(\secp)}$ as $G^{\xor s}(k_1,\dots,k_s) := \bigxor_{i=1}^s G(k_i)$. Then $G^{\xor s}$ is a sQPRG with pseudodeterminism $1 - O(\secp^{-(c-1)})$ and output length $\ell(\secp)$.
\end{theorem}
\begin{proof}[Proof of Strong Security] \hfill

\begin{enumerate}
    \item By ``pseudorandomness implies next-bit ununpredictability'' in~\Cref{lemma:next-bit}, for a random seed $k$ and any $i\in\bit^{\ell(\secp)}$, the probability of outputting $y_i$ given $y_{[1:i-1]}$ is at most $1/2 + \delta$ (this is where we need that $\delta(\secp) \le 0.49+o(1)$), where $y\gets G(k)$.
    \item By~\Cref{lemma:Bit-wise direct product lemma}, for any $i\in\bit^{\ell(\secp)}$, the probability of computing $y^1_{i+1}||\dots||y^{2s}_{i+1}$ from $y^1_{[1:i]}||\dots$ $||y^{2s}_{[1:i]}$ for independent, uniform seeds $k_1,\dots,k_{2s}$ is at most $\veps = e^{-\Omega(s)} = \negl(\secp)$.
    \item By~\Cref{lemma:Direct product theorem implies xor lemma}, for any $i\in\bit^{\ell(\secp)}$, the probability of computing the XOR of $y^1_{i+1}, \dots, y^{s}_{i+1}$ (the $(i+1)$-th bit of $G^{\xor s}(k_1, \dots k_s)$) from $y^1_{[1:i]}\xor \dots \xor y^{s}_{[1:i]}$ (the first $i$ bits of $G^{\xor s}(k_1, \dots k_s)$) for independent, uniform seeds $k_1,\dots,k_{2s}$ is at most $1/2 + \poly(s\veps) = 1/2 + \negl(\secp)$.
    \item By ``next-bit ununpredictability implies pseudorandomness'' in~\Cref{lemma:next-bit}, we can conclude that $G^{\xor s}$ is $\left(1 - \ell(\secp)\cdot \poly(s\veps)\right)$-pseudorandom, where $\ell(\secp)\cdot \poly(s\veps) = \negl(\secp)$.
\end{enumerate}
\end{proof}

%% file: tcs.bib
@misc{BBSS23eprint,
      author = {Amit Behera and Zvika Brakerski and Or Sattath and Omri Shmueli},
      title = {Pseudorandomness with Proof of Destruction and Applications},
      howpublished = {Cryptology ePrint Archive, Paper 2023/543},
      year = {2023},
      note = {\url{https://eprint.iacr.org/2023/543}},
      url = {https://eprint.iacr.org/2023/543}
}

@inproceedings{GJMZ23,
  title={Commitments to quantum states},
  author={Gunn, Sam and Ju, Nathan and Ma, Fermi and Zhandry, Mark},
  booktitle={Proceedings of the 55th Annual ACM Symposium on Theory of Computing},
  pages={1579--1588},
  year={2023}
}

@inproceedings{Gavinsky12,
  title={Quantum money with classical verification},
  author={Gavinsky, Dmitry},
  booktitle={2012 IEEE 27th Conference on Computational Complexity},
  pages={42--52},
  year={2012},
  organization={IEEE}
}

@inproceedings{ODW14,
  title={Goldreich's PRG: evidence for near-optimal polynomial stretch},
  author={ODonnell, Ryan and Witmer, David},
  booktitle={2014 IEEE 29th Conference on Computational Complexity (CCC)},
  pages={1--12},
  year={2014},
  organization={IEEE}
}

@inproceedings{impagliazzo1997p,
  title={P= BPP if E requires exponential circuits: Derandomizing the XOR lemma},
  author={Impagliazzo, Russell and Wigderson, Avi},
  booktitle={Proceedings of the twenty-ninth annual ACM symposium on Theory of computing},
  pages={220--229},
  year={1997}
}

@inproceedings{Naor89CRTPYO,
  title={Bit Commitment Using Pseudo-Randomness},
  booktitle={Advances in Cryptology - CRYPTO '89, 9th Annual International Cryptology Conference, Santa Barbara, California, USA, August 20-24, 1989, Proceedings},
  series={Lecture Notes in Computer Science},
  publisher={Springer},
  volume={435},
  pages={128-136},
  doi={10.1007/0-387-34805-0_13},
  author={Moni Naor},
  year=1989
}

@inproceedings{LM13CRYPTO,
  title={Man-in-the-Middle Secure Authentication Schemes from LPN and Weak PRFs},
  booktitle={CRYPTO},
  publisher={Springer},
  pages={308-325},
  url={https://www.iacr.org/archive/crypto2013/80420277/80420277.pdf},
  doi={10.1007/978-3-642-40084-1_18},
  author={Vadim Lyubashevsky and Daniel Masny},
  year=2013
}

@inproceedings{DKPW12Eurocrypt,
  title={Message Authentication, Revisited},
  booktitle={EUROCRYPT},
  series={Lecture Notes in Computer Science},
  publisher={Springer},
  volume={7237},
  pages={355-374},
  url={https://www.iacr.org/archive/eurocrypt2012/72370349/72370349.pdf},
  doi={10.1007/978-3-642-29011-4_22},
  author={Yevgeniy Dodis and Eike Kiltz and Krzysztof Pietrzak and Daniel Wichs},
  year=2012
}

@inproceedings{Pie09Eurocrypt,
  title={A Leakage-Resilient Mode of Operation},
  booktitle={Advances in Cryptology - EUROCRYPT 2009, 28th Annual International Conference on the Theory and Applications of Cryptographic Techniques},
  series={Lecture Notes in Computer Science},
  publisher={Springer},
  volume={5479},
  pages={462-482},
  doi={10.1007/978-3-642-01001-9_27},
  author={Krzysztof Pietrzak},
  year=2009
}

@inproceedings{MS07Eurocrypt,
  title={A Fast and Key-Efficient Reduction of Chosen-Ciphertext to Known-Plaintext Security},
  booktitle={Advances in Cryptology - EUROCRYPT 2007, 26th Annual International Conference on the Theory and Applications of Cryptographic Techniques, Barcelona, Spain, May 20-24, 2007, Proceedings},
  series={Lecture Notes in Computer Science},
  publisher={Springer},
  volume={4515},
  pages={498-516},
  url={https://iacr.org/archive/eurocrypt2007/45150498/45150498.pdf},
  doi={10.1007/978-3-540-72540-4_29},
  author={Ueli Maurer and Johan Sjödin},
  year=2007
}

@inproceedings{DN02CRYPTO,
  title={Expanding Pseudorandom Functions; or: From Known-Plaintext Security to Chosen-Plaintext Security},
  booktitle={Advances in Cryptology - CRYPTO 2002, 22nd Annual International Cryptology Conference, Santa Barbara, California, USA, August 18-22, 2002, Proceedings},
  series={Lecture Notes in Computer Science},
  publisher={Springer},
  volume={2442},
  pages={449-464},
  url={http://www.iacr.org/cryptodb/archive/2002/CRYPTO/1220/1220.ps},
  doi={10.1007/3-540-45708-9_29},
  author={Ivan Damgård and Jesper Buus Nielsen},
  year=2002
}

@article{DCEL09PRA,
  title = {Exact and approximate unitary 2-designs and their application to fidelity estimation},
  author = {Dankert, Christoph and Cleve, Richard and Emerson, Joseph and Livine, Etera},
  journal = {Phys. Rev. A},
  volume = {80},
  issue = {1},
  pages = {012304},
  numpages = {6},
  year = {2009},
  month = {Jul},
  publisher = {American Physical Society},
  doi = {10.1103/PhysRevA.80.012304},
  url = {https://link.aps.org/doi/10.1103/PhysRevA.80.012304}
}

@inproceedings{AC02,
  title={A quantum Goldreich-Levin theorem with cryptographic applications},
  author={Adcock, Mark and Cleve, Richard},
  booktitle={STACS 2002: 19th Annual Symposium on Theoretical Aspects of Computer Science Antibes-Juan les Pins, France, March 14--16, 2002 Proceedings},
  pages={323--334},
  year={2002},
  organization={Springer}
}

@INPROCEEDINGS{AE07CCC,
  author={Ambainis, Andris and Emerson, Joseph},
  booktitle={Twenty-Second Annual IEEE Conference on Computational Complexity (CCC'07)}, 
  title={Quantum t-designs: t-wise Independence in the Quantum World}, 
  year={2007},
  volume={},
  number={},
  pages={129-140},
  doi={10.1109/CCC.2007.26}}

@inproceedings{ABGKR14ITCS,
  title={Candidate weak pseudorandom functions in $\mathsf{AC}^0$ $\circ$ $\mathsf{Mod}_2$},
  author={Akavia, Adi and Bogdanov, Andrej and Guo, Siyao and Kamath, Akshay and Rosen, Alon},
  booktitle={Proceedings of the 5th conference on Innovations in theoretical computer science},
  pages={251--260},
  year={2014}
}

@INPROCEEDINGS{Yao82FOCS,
  author={Yao, Andrew C.},
  booktitle={23rd Annual Symposium on Foundations of Computer Science (sfcs 1982)}, 
  title={Theory and application of trapdoor functions}, 
  year={1982},
  volume={},
  number={},
  pages={80-91},
  doi={10.1109/SFCS.1982.45}}

@inproceedings{RS19,
  title={Semi-quantum money},
  author={Radian, Roy and Sattath, Or},
  booktitle={Proceedings of the 1st ACM Conference on Advances in Financial Technologies},
  pages={132--146},
  year={2019}
}

@article{Mul59,
    author = {Muller, Mervin E.},
    title = {A Note on a Method for Generating Points Uniformly on N-Dimensional Spheres},
    year = {1959},
    issue_date = {April 1959},
    publisher = {Association for Computing Machinery},
    address = {New York, NY, USA},
    volume = {2},
    number = {4},
    issn = {0001-0782},
    url = {https://doi.org/10.1145/377939.377946},
    doi = {10.1145/377939.377946},
    journal = {Commun. ACM},
    month = {apr},
    pages = {19–20},
    numpages = {2}
}

@article{DD87Haar,
    author = {Diaconis, Persi and Freedman, David},
    title = {A dozen de {Finetti-style} results in search of a theory},
    journal = {Annales de l'I.H.P. Probabilit\'es et statistiques},
    pages = {397--423},
    publisher = {Gauthier-Villars},
    volume = {23},
    number = {S2},
    year = {1987},
    zbl = {0619.60039},
    mrnumber = {898502},
    language = {en},
    url = {http://www.numdam.org/item/AIHPB_1987__23_S2_397_0/}
}

@inproceedings{bouland2019computational,
  author    = {Adam Bouland and
               Bill Fefferman and
               Umesh V. Vazirani},
  editor    = {Thomas Vidick},
  title     = {Computational Pseudorandomness, the Wormhole Growth Paradox, and Constraints
               on the AdS/CFT Duality (Abstract)},
  booktitle = {11th Innovations in Theoretical Computer Science Conference, {ITCS}
               2020, January 12-14, 2020, Seattle, Washington, {USA}},
  series    = {LIPIcs},
  volume    = {151},
  pages     = {63:1--63:2},
  publisher = {Schloss Dagstuhl - Leibniz-Zentrum f{\"{u}}r Informatik},
  year      = {2020},
  doi       = {10.4230/LIPIcs.ITCS.2020.63},
  bibsource = {dblp computer science bibliography, https://dblp.org}
}

@misc{MY21,
  doi = {10.48550/ARXIV.2112.06369},
  
  author = {Morimae, Tomoyuki and Yamakawa, Takashi},
  
  title = {Quantum commitments and signatures without one-way functions},
  
  publisher = {arXiv},
  
  year = {2021},
  
  copyright = {Creative Commons Attribution 4.0 International}
}

@article{HBC+22,
  title={Quantum advantage in learning from experiments},
  author={Huang, Hsin-Yuan and Broughton, Michael and Cotler, Jordan and Chen, Sitan and Li, Jerry and Mohseni, Masoud and Neven, Hartmut and Babbush, Ryan and Kueng, Richard and Preskill, John and others},
  journal={Science},
  volume={376},
  number={6598},
  pages={1182--1186},
  year={2022},
  publisher={American Association for the Advancement of Science}
}

@article{BFGVZ22,
  title={Quantum Pseudoentanglement},
  author={Bouland, Adam and Fefferman, Bill and Ghosh, Soumik and Vazirani, Umesh and Zhou, Zixin},
  journal={arXiv preprint arXiv:2211.00747},
  year={2022}
}

@misc{MY22onewayness,
      author = {Tomoyuki Morimae and Takashi Yamakawa},
      title = {One-Wayness in Quantum Cryptography},
      howpublished = {Cryptology ePrint Archive, Paper 2022/1336},
      year = {2022},
      note = {\url{https://eprint.iacr.org/2022/1336}},
      url = {https://eprint.iacr.org/2022/1336}
}

@inproceedings{CHS05TCC,
  title     = {Hardness amplification of weakly verifiable puzzles},
  author    = {Canetti, Ran and Halevi, Shai and Steiner, Michael},
  booktitle = {Theory of Cryptography: Second Theory of Cryptography Conference, TCC 2005, Cambridge, MA, USA, February 10-12, 2005. Proceedings 2},
  pages     = {17--33},
  year      = {2005},
  publisher = {Springer}
}

@article{IJK08JOC,
  title     = {Chernoff-type direct product theorems},
  author    = {Impagliazzo, Russell and Jaiswal, Ragesh and Kabanets, Valentine},
  journal   = {Journal of Cryptology},
  volume    = {22},
  number    = {1},
  pages     = {75--92},
  year      = {2009},
  publisher = {Springer}
}

@inproceedings{DIJK09TCC,
  title     = {Security amplification for interactive cryptographic primitives},
  author    = {Dodis, Yevgeniy and Impagliazzo, Russell and Jaiswal, Ragesh and Kabanets, Valentine},
  booktitle = {Theory of Cryptography: 6th Theory of Cryptography Conference, TCC 2009, San Francisco, CA, USA, March 15-17, 2009. Proceedings 6},
  pages     = {128--145},
  year      = {2009},
  organization={Springer}
}

@inproceedings{MauTes09CRYPTO,
    author       = {Ueli Maurer and Stefano Tessaro},
    title        = {Computational Indistinguishability Amplification: Tight Product Theorems for System Composition},
    editor       = {Shai Halevi},
    booktitle    = {Advances in Cryptology --- CRYPTO 2009},
    pages        = {350--368},
    series       = {Lecture Notes in Computer Science},
    volume       = {5677},
    year         = {2009},
    month        = {8},
    publisher    = {Springer-Verlag},
}

@inproceedings{MauTes10TCC,
    author       = {Ueli Maurer and Stefano Tessaro},
    title        = {A Hardcore Lemma for Computational   Indistinguishability: Security Amplification for Arbitrarily Weak PRGs with Optimal Stretch},
    editor       = {Daniele Micciancio},
    booktitle    = {Theory of Cryptography --- TCC 2010},
    pages        = {237--254},
    series       = {Lecture Notes in Computer Science},
    volume       = {5978},
    year         = {2010},
    month        = {2},
    publisher    = {Springer-Verlag},
}

@article{sirazhdinov1962convergence,
  title={On convergence in the mean for densities},
  author={Sirazhdinov, S Kh and Mamatov, M},
  journal={Theory of Probability \& Its Applications},
  volume={7},
  number={4},
  pages={424--428},
  year={1962},
  publisher={SIAM}
}

@book{Meckes19,
  title={The random matrix theory of the classical compact groups},
  author={Meckes, Elizabeth S},
  volume={218},
  year={2019},
  publisher={Cambridge University Press}
}

@inproceedings{AGQY22,
  title={Pseudorandom (Function-Like) Quantum State Generators: New Definitions and Applications},
  author={Ananth, Prabhanjan and Gulati, Aditya and Qian, Luowen and Yuen, Henry},
  booktitle={Theory of Cryptography Conference},
  pages={237--265},
  year={2022},
  organization={Springer}
}

@article{Naor91,
author={Naor, Moni},
title={Bit commitment using pseudorandomness},
journal={Journal of Cryptology},
year={1991},
month={1},
day={01},
volume={4},
number={2},
pages={151-158},
issn={1432-1378},
doi={10.1007/BF00196774}
}

@inproceedings{JLS18,
  author    = {Zhengfeng Ji and
               Yi{-}Kai Liu and
               Fang Song},
  editor    = {Hovav Shacham and
               Alexandra Boldyreva},
  title     = {Pseudorandom Quantum States},
  booktitle = {Advances in Cryptology - {CRYPTO} 2018 - 38th Annual International
               Cryptology Conference, Santa Barbara, CA, USA, August 19-23, 2018,
               Proceedings, Part {III}},
  series    = {Lecture Notes in Computer Science},
  volume    = {10993},
  pages     = {126--152},
  publisher = {Springer},
  year      = {2018},
  doi       = {10.1007/978-3-319-96878-0_5},
  bibsource = {dblp computer science bibliography, https://dblp.org}
}

@book{nielsen_chuang_2010, place={Cambridge}, title={Quantum Computation and Quantum Information: 10th Anniversary Edition}, DOI={10.1017/CBO9780511976667}, publisher={Cambridge University Press}, author={Nielsen, Michael A. and Chuang, Isaac L.}, year={2010}}

@inproceedings{RR94,
  title={Natural proofs},
  author={Razborov, Alexander A and Rudich, Steven},
  booktitle={Proceedings of the twenty-sixth annual ACM symposium on Theory of computing},
  pages={204--213},
  year={1994}
}

@article{NW94,
  title={Hardness vs randomness},
  author={Nisan, Noam and Wigderson, Avi},
  journal={Journal of computer and System Sciences},
  volume={49},
  number={2},
  pages={149--167},
  year={1994},
  publisher={Elsevier}
}

@article{AIK06,
  title={Computationally private randomizing polynomials and their applications},
  author={Applebaum, Benny and Ishai, Yuval and Kushilevitz, Eyal},
  journal={computational complexity},
  volume={15},
  number={2},
  pages={115--162},
  year={2006},
  publisher={Springer}
}

@article{Gol90,
title = {A note on computational indistinguishability},
journal = {Information Processing Letters},
volume = {34},
number = {6},
pages = {277-281},
year = {1990},
issn = {0020-0190},
doi = {10.1016/0020-0190(90)90010-U},
author = {Oded Goldreich}
}

@inproceedings{AQY21,
      title={Cryptography from Pseudorandom Quantum States.}, 
      author={Ananth, Prabhanjan and Qian, Luowen and Yuen, Henry},
      booktitle={CRYPTO},
      year={2022}
}

@inproceedings{GLSV21,
  author    = {Alex B. Grilo and
               Huijia Lin and
               Fang Song and
               Vinod Vaikuntanathan},
  editor    = {Anne Canteaut and
               Fran{\c{c}}ois{-}Xavier Standaert},
  title     = {Oblivious Transfer Is in MiniQCrypt},
  booktitle = {Advances in Cryptology - {EUROCRYPT} 2021 - 40th Annual International
               Conference on the Theory and Applications of Cryptographic Techniques,
               Zagreb, Croatia, October 17-21, 2021, Proceedings, Part {II}},
  series    = {Lecture Notes in Computer Science},
  volume    = {12697},
  pages     = {531--561},
  publisher = {Springer},
  year      = {2021},
  doi       = {10.1007/978-3-030-77886-6_18},
  bibsource = {dblp computer science bibliography, https://dblp.org}
}

@article{hayden2006aspects,
author={Hayden, Patrick
and Leung, Debbie W.
and Winter, Andreas},
title={Aspects of Generic Entanglement},
journal={Communications in Mathematical Physics},
year={2006},
month={7},
day={01},
volume={265},
number={1},
pages={95-117},
issn={1432-0916},
doi={10.1007/s00220-006-1535-6}
}

@inproceedings{BrakerskiS20,
  author    = {Zvika Brakerski and
               Omri Shmueli},
  editor    = {Daniele Micciancio and
               Thomas Ristenpart},
  title     = {Scalable Pseudorandom Quantum States},
  booktitle = {Advances in Cryptology - {CRYPTO} 2020 - 40th Annual International
               Cryptology Conference, {CRYPTO} 2020, Santa Barbara, CA, USA, August
               17-21, 2020, Proceedings, Part {II}},
  series    = {Lecture Notes in Computer Science},
  volume    = {12171},
  pages     = {417--440},
  publisher = {Springer},
  year      = {2020},
  doi       = {10.1007/978-3-030-56880-1_15},
  bibsource = {dblp computer science bibliography, https://dblp.org}
}

@inproceedings{LP20,
  title={On one-way functions and Kolmogorov complexity},
  author={Liu, Yanyi and Pass, Rafael},
  booktitle={2020 IEEE 61st Annual Symposium on Foundations of Computer Science (FOCS)},
  pages={1243--1254},
  year={2020},
  organization={IEEE}
}

@inproceedings{BY22,
  title={Quantum garbled circuits},
  author={Brakerski, Zvika and Yuen, Henry},
  booktitle={Proceedings of the 54th Annual ACM SIGACT Symposium on Theory of Computing},
  pages={804--817},
  year={2022}
}

@inproceedings{Kretschmer21,
  author    = {William Kretschmer},
  editor    = {Min{-}Hsiu Hsieh},
  title     = {Quantum Pseudorandomness and Classical Complexity},
  booktitle = {16th Conference on the Theory of Quantum Computation, Communication
               and Cryptography, {TQC} 2021, July 5-8, 2021, Virtual Conference},
  series    = {LIPIcs},
  volume    = {197},
  pages     = {2:1--2:20},
  publisher = {Schloss Dagstuhl - Leibniz-Zentrum f{\"{u}}r Informatik},
  year      = {2021},
  doi       = {10.4230/LIPIcs.TQC.2021.2},
  bibsource = {dblp computer science bibliography, https://dblp.org}
}

@inproceedings{BartusekCKM21a,
  author    = {James Bartusek and
               Andrea Coladangelo and
               Dakshita Khurana and
               Fermi Ma},
  editor    = {Tal Malkin and
               Chris Peikert},
  title     = {One-Way Functions Imply Secure Computation in a Quantum World},
  booktitle = {Advances in Cryptology - {CRYPTO} 2021 - 41st Annual International
               Cryptology Conference, {CRYPTO} 2021, Virtual Event, August 16-20,
               2021, Proceedings, Part {I}},
  series    = {Lecture Notes in Computer Science},
  volume    = {12825},
  pages     = {467--496},
  publisher = {Springer},
  year      = {2021},
  doi       = {10.1007/978-3-030-84242-0_17},
  bibsource = {dblp computer science bibliography, https://dblp.org}
}

@article{GGM86,
  title={How to construct random functions},
  author={Goldreich, Oded and Goldwasser, Shafi and Micali, Silvio},
  journal={Journal of the ACM (JACM)},
  volume={33},
  number={4},
  pages={792--807},
  year={1986},
  publisher={ACM New York, NY, USA}
}
